\documentclass[10pt,conference,twocolumn]{IEEEtran}

\usepackage[utf8]{inputenc}

\usepackage{graphicx}
\usepackage{amsmath, bm}
\usepackage{amsfonts}
\usepackage{amssymb}

\usepackage{cite}

\usepackage{multirow}

\usepackage{url}
\usepackage{balance}
\usepackage{xcolor}
\usepackage[capitalise,nameinlink,noabbrev]{cleveref}
\usepackage[font=small,skip=0pt]{caption}
\usepackage{subcaption}
\usepackage{tikz}
\date{}

\newcommand{\comment}[1]{}
\newtheorem{theorem}{Theorem}
\newtheorem{lemma}{Lemma}

\newtheorem{corollary}{Corollary}
\newtheorem{definition}{Definition}
\newtheorem{remark}{Remark}

\newtheorem{claim}{Claim}
\begin{document}
	\IEEEoverridecommandlockouts
\title{Heterogeneous Demand-Private Coded Caching: Exact Characterization for Two Users and Two Files }
\author{
	\IEEEauthorblockN{Tati V D N Eswar and Jithin~Ravi}
	\IEEEauthorblockA{
				Department of Electronics and Electrical Communication Engineering\\
		Indian Institute of Technology Kharagpur, India\\
		%	\\
		Emails: {eswar27.24@kgpian.iitkgp.ac.in, jithin@ece.iitkgp.ac.in}}
	\thanks{J.~Ravi has  received funding from IIT Kharagpur under a Faculty Start-up Research Grant (FSRG).}
}
	\maketitle
	\begin{abstract}
We consider the demand-private coded caching problem with heterogeneous file sizes and cache capacities, where each user must recover its requested file without learning any information about the demands of the other users. We first extend the virtual-user reduction technique to accommodate heterogeneous cache sizes, thereby enabling the construction of demand-private coded caching schemes from non-private schemes. Our main contribution is an exact characterization of the memory–rate trade-off for the fundamental setting of two users and two files with arbitrary file sizes and cache capacities. Specifically, we characterize the optimal trade-off through eight linear inequalities and establish achievability by constructing explicit coded caching schemes for the non-trivial corner points, followed by memory sharing. These constructions include novel asymmetric cache placement strategies tailored to heterogeneous file sizes and cache capacities. For the converse, we exploit the distributional equivalences induced by demand privacy to derive information-theoretic lower bounds that explicitly account for the underlying heterogeneity. Our results generalize the exact memory–rate trade-off characterization for the homogeneous demand-private coded caching setting with two users and two files, while quantifying the additional transmission cost imposed by demand privacy compared with the non-private heterogeneous setting.
\end{abstract}

\section{Introduction}
\label{sec:intro}

In the seminal work~\cite{MaddahAli14}, Maddah-Ali and Niesen
introduced the coded caching framework for a noiseless broadcast
network consisting of a server with $N$ files of equal size and $K$
users, each equipped with a cache of equal size. The system operates
in two phases: a \emph{placement phase}, in which the user caches are
populated without knowledge of future demands, and a \emph{delivery
phase}, in which each user requests one of the files and the server
broadcasts a common coded message to satisfy all requests. By jointly
designing the placement and delivery phases, coded caching creates
multicast opportunities across different demands and achieves a
\emph{global caching gain} in addition to the conventional
\emph{local caching gain}. The scheme in~\cite{MaddahAli14} is
order-optimal within a constant factor for general system parameters,
and the exact memory--rate trade-off was characterized for the
fundamental $N=K=2$ setting. The memory-rate trade-off was subsequently improved and the constant factor in the order-optimality was  tightened~\cite{Yu18, Yu19, Ghasemi17, Wang18}.

Several  works have extended the coded caching framework
to different network models, including hierarchical
networks~\cite{Karamchandani16} and multi-access
networks~\cite{Reddy20,Namboodiri24}. Another important line of work
addresses security~\cite{Sengupta15} and privacy concerns~\cite{Ravindrakumar18} in coded caching. In
particular, the coded transmissions in the original framework may
reveal the users' demands to one another, motivating the study of
\emph{demand-private coded caching}~\cite{Wan21,  Gurjarpadhye22, Aravind20}. Under information-theoretic
demand privacy, each user must recover its requested file while
learning no information about the demands of the other users.
The exact
demand-private memory-rate trade-off was  characterized for the $N=K=2$ setting~\cite{KamathRD20}.

Most existing works on coded caching assume homogeneous file sizes and cache capacities. In practical scenarios, however, files may have different sizes, users may have unequal cache capacities, and user demands may follow non-uniform distributions. Various forms of such heterogeneity have been investigated in~\cite{Daniel20,Deng22,Li17,Chang22,Bose25}. In particular, Chang \emph{et al.}~\cite{Li17,Chang22} studied the \emph{fully heterogeneous} coded caching problem, which simultaneously accommodates heterogeneous file sizes, cache capacities, and arbitrary demand distributions, in the non-private setting. For the special case of $N=K=2$, an exact characterization of the memory-rate trade-off was obtained in~\cite{Chang22}. Despite these advances, the coded caching problem under full heterogeneity with information-theoretic demand privacy remains unexplored.

In this paper, we investigate demand-private coded caching under full heterogeneity, encompassing (i) heterogeneous file sizes, (ii) heterogeneous cache sizes, and (iii) arbitrary user demand distributions. We first establish that the demand privacy constraint necessitates a transmission length that is invariant across all demand vectors, rendering the transmission rate independent of the underlying demand distributions (Remark~\ref{rem:avg_max_rate}). Consequently, the average and worst-case transmission rates coincide, and it suffices to characterize the worst-case memory--rate trade-off. For achievability, we extend the virtual-user technique and the restricted-demand-subset reduction, previously developed for homogeneous demand-private coded caching, to the fully heterogeneous setting.

Our main result, \Cref{thm:exact_tradeoff}, exactly characterizes the demand-private memory-rate trade-off for $N=K=2$ with arbitrary file sizes and cache capacities through eight linear inequalities. This characterization generalizes the homogeneous demand-private result~\cite{Gurjarpadhye22} and quantifies the rate penalty imposed by demand privacy relative to the non-private heterogeneous setting~\cite{Chang22}. To establish achievability, we extend the virtual-user technique and restricted-demand-subset reduction to the heterogeneous setting and construct four non-trivial schemes. These include two novel asymmetric cache placement schemes for unequal cache capacities and two extensions of existing homogeneous-cache constructions to heterogeneous file sizes. Together with memory sharing, these schemes achieve all corner points of the optimal region. For the converse, we exploit distributional equivalences induced by demand privacy to derive eight information-theoretic bounds, including four new bounds specific to heterogeneity. The matching achievability and converse establish the optimality of the proposed schemes.

The remainder of the paper is organized as follows. In \Cref{sec:model}, we introduce the heterogeneous demand-private coded caching setting and formulate the problem. In \Cref{sec:results}, we state our main results, followed by their proofs in\Cref{sec:proofs}.   In \Cref{sec:discussion}, we discuss the implications of our findings. Additional proofs are provided in the appendices.

\section{System Model and Problem Formulation}
\label{sec:model}

Consider a server with access to $N$ independent files
$W_0,W_1,\ldots,W_{N-1}$, where file $W_i$ has size $F_i$ bits. Specifically, $W_i$ is uniformly distributed over\footnote{We use the notation $[n]$ to denote the set $\{0,1,\ldots,n-1\}$.}
\begin{align}
[2^{F_i}] \triangleq \{0,1,\ldots,2^{F_i}-1\}.
\end{align}
The server is connected to $K$ users through a noiseless broadcast link. User $k\in[K]$ has a local cache of size $M_k$ bits, where the cache sizes may differ across users. Thus, the model accommodates both heterogeneous file sizes and heterogeneous cache sizes.

The system operates in two phases: the \emph{placement phase} and the \emph{delivery phase}. During the placement phase, the server populates the users' caches without knowledge of their future demands. In the delivery phase, each user $k$ requests a file $W_{D_k}$, where $D_k\in[N]$. Let
\begin{align}
    \overline D=(D_0,D_1,\ldots,D_{K-1})
\end{align}
denote the demand vector, and let
\begin{align}
    \overline D_{-k}
    =
    (D_0,\ldots,D_{k-1},D_{k+1},\ldots,D_{K-1}),
    \quad k\in[K],
\end{align}
denote the demands of all users other than user $k$. We assume that the users' demands are mutually independent, but not necessarily identically distributed\footnote{Correlated demands may inherently reveal information about other users' demands, which is excluded from our setting.}. The demand vector $\overline D$ follows a prescribed distribution $p_{\overline D}$ on $[N]^K$, which is known to the server.

All users communicate their demands to the server. Upon receiving the demand vector $\overline D$, the server broadcasts a message $X_{\overline D}$ to all users over the noiseless broadcast link. The size of the transmitted message may depend on the demand vector and is denoted by $R_{\overline D}$ bits. Using its cache content, its own demand, and the received transmission, each user must recover its requested file.

In addition to the recovery requirement, we impose a \emph{demand-privacy} constraint: each user must learn no information about the demands of the other users. To facilitate demand privacy, we allow the server and user $k$ to share randomness $S_k$ that is unknown to all other users. Such shared randomness can be generated and securely placed during the placement phase. Let $S_0,\ldots,S_{K-1}$ take values in finite alphabets $\mathcal S_0,\ldots,\mathcal S_{K-1}$, respectively, and define
\begin{align}
    \overline S=(S_0,\ldots,S_{K-1}).
\end{align}
We assume that the collections
\begin{align}
\{S_k:k\in[K]\},\quad
\{D_k:k\in[K]\},\quad
\{W_i:i\in[N]\}
\end{align}
are mutually independent.

A heterogeneous demand-private coded caching scheme consists of $K$ cache encoding functions, one transmission encoding function for each possible demand vector, and $K$ decoding functions. Specifically, for each $k\in[K]$, the cache encoding function is
\begin{align}
    C_k:
    \mathcal S_k
    \times
    \prod_{i=0}^{N-1}[2^{F_i}]
    \longrightarrow
    [2^{M_k}],
\end{align}
and the cache content available to user $k$ is
\begin{align}
    Z_k=\big(C_k(S_k,\overline W),S_k\big),
\end{align}
where
\begin{align}
    \overline W=(W_0,\ldots,W_{N-1}).
\end{align}

For each demand vector $\overline d\in[N]^K$, the transmission encoding function is
\begin{align}
    E_{\overline d}:
    \prod_{i=0}^{N-1}[2^{F_i}]
    \times
    \prod_{k=0}^{K-1}\mathcal S_k
    \longrightarrow
    [2^{R_{\overline d}}],
\end{align}
and the corresponding transmitted message is
\begin{align}
    X_{\overline d}
    =
    E_{\overline d}(\overline W,\overline S).
\end{align}

For each $k\in[K]$, the decoding function is
\begin{align}
    G_k:
    [N]\times\mathcal S_k
    \times[2^{R_{\overline D}}]
    \times[2^{M_k}]
    \longrightarrow
    [2^{F_{D_k}}],
\end{align}
and produces the estimate
\begin{align}
    \widehat W_{D_k}
    =
    G_k
    \big(
    D_k,S_k,X_{\overline D},C_k(S_k,\overline W)
    \big).
\end{align}

Define
\begin{align}
    \mathbf M & \triangleq (M_0,\ldots,M_{K-1}),\\
    \mathbf F & \triangleq (F_0,\ldots,F_{N-1}),\\
    \mathcal C& \triangleq \{C_k:k\in[K]\},\\
    \mathcal G& \triangleq\{G_k:k\in[K]\},\\
    \mathcal E& \triangleq\{E_{\overline d}:\overline d\in[N]^K\}.
\end{align}
The collection
\begin{align}
    (\mathcal C,\mathcal E,\mathcal G)
\end{align}
is called an $(N,K,\mathbf M,\mathbf F)$ heterogeneous  demand-private coded caching scheme if it satisfies the following recovery and demand-privacy conditions. In particular, we consider the zero-error recovery condition i.e., 
\begin{align}
\label{Eq_dec_cond_hetero}
    H\left(
    W_{D_k}
    \middle|
    Z_k,D_k,X_{\overline D}
    \right)=0,
    \qquad \forall k\in[K].
\end{align}
The demand-privacy condition requires that user $k$ obtain no information about the demands of the other users, i.e.,
\begin{align}
\label{Eq_priv_cond_hetero}
    I\left(
    \overline D_{-k};
    Z_k,D_k,X_{\overline D}
    \right)=0,
    \qquad \forall k\in[K].
\end{align}

For a given demand vector $\overline d$, let $R_{\overline d}$ denote the number of bits transmitted by the server. The average transmission rate of a scheme is defined as
\begin{align}
    R_{\mathrm{avg}}
    =
    \sum_{\overline d\in[N]^K}
    p_{\overline D}(\overline d)R_{\overline d}.
\label{Eq_avg_rate_hetero}
\end{align}
For worst-case analysis, the transmission rate is defined as
\begin{align}
    R_{\max}
    =
    \max_{\overline d\in[N]^K}R_{\overline d}.
\label{Eq_wc_rate_hetero}
\end{align}

We have the following remark.
\begin{remark}
\label{rem:avg_max_rate}
Consider an $(N,K,\mathbf{M},\mathbf{F})$ heterogeneous demand-private coded caching scheme satisfying \eqref{Eq_priv_cond_hetero}. The demand-privacy constraint in \eqref{Eq_priv_cond_hetero}. Then, we have
\begin{align}
    I\left(
    \overline D; X_{\overline D}
    \right)= &  I\left( D_k;X_{\overline D} 
    \right) +  I\left(
    \overline D_{-k};X_{\overline D}| D_k 
    \right) \notag\\
    = & I\left(
    \overline D_{-k};X_{\overline D}| D_k 
    \right) \label{eq:Tx_demn_1}\\
    = & 0. \label{eq:Tx_demn_2}
\end{align}
where \eqref{eq:Tx_demn_1}  follows from the demand privacy condition for user $i\neq k$ and \eqref{eq:Tx_demn_2} follows from the demand privacy condition for user $k$. Thus, $X_{\overline D}$ is independent of the
entire demand vector $\overline D$. In particular,
for any $\overline d_1,\overline d_2\in[N]^K$,
\begin{align}
&H(X_{\overline D}\mid
  \overline D=\overline d_1)
\nonumber =
H(X_{\overline D}\mid
  \overline D=\overline d_2).
\end{align}
 Consequently, without loss of generality (w.l.o.g.), it suffices to consider the worst-case transmission rate. In the remainder of the paper, we assume that the demand vector $\overline D$ is uniformly distributed over $[N]^K$ and denote $R_{\max}$ simply by $R$.
\end{remark}

For given $N$, $K$, $\mathbf{M}$, and $\mathbf{F}$, a transmission rate $R$ is said to be \emph{achievable with demand privacy} if there exists an $(N,K,\mathbf{M},\mathbf{F})$ heterogeneous demand-private coded caching scheme satisfying the decodability and demand-privacy constraints in \eqref{Eq_dec_cond_hetero} and \eqref{Eq_priv_cond_hetero}, respectively, with transmission rate $R$. We denote such a scheme by $(N,K,\mathbf{M},\mathbf{F},R)$. Our goal in this paper is to characterize the set of all achievable tuples $(N,K,\mathbf{M},\mathbf{F},R)$.

\comment{
The corresponding optimal memory-rate trade-off is defined as
\begin{align}
    &R^{*p}_{N,K}(\mathbf M,\mathbf  F) \nonumber \\
    & \;\triangleq
    \inf
    \left\{
    R:
    R\text{ is achievable with demand privacy}
    \right\}.
\label{Eq_opt_rate_priv_hetero}
\end{align}

Similarly, by removing the shared randomness and the demand-privacy constraint, we obtain the corresponding non-private heterogeneous coded caching problem. In this case, the cache encoding, transmission encoding, and decoding functions do not depend on any shared randomness. The optimal worst-case transmission rate is denoted by
\begin{align}
    R^{*}_{N,K}(\mathbf M,\mathbf F),
\end{align}
where the demand vector is assumed to be uniformly distributed over $[N]^K$.
}
\label{sec:results}

\section{Results}
\label{sec:results}
%===========================================================
%  results.tex (Top portion up to Non-trivial schemes)
%===========================================================
\providecommand{\FA}{F_0}
\providecommand{\FB}{F_1}
\providecommand{\Arem}{A_{\text{rem}}}
\providecommand{\Acore}{A_{\text{core}}}
% Bold run-in heading inside proofs
\providecommand{\step}[1]{\par\medskip\noindent\textbf{#1}\quad}
% Both demands written compactly:  \dem{a}{b}  ->  D_0=a, D_1=b
\providecommand{\dem}[2]{D_0{=}#1,\,D_1{=}#2}

% Local layout settings; the group keeps them from leaking into the
% appendix and the rest of the paper.
\begingroup
\setlength{\jot}{4pt}
\allowdisplaybreaks

%===========================================================
We first present a general construction that transforms a
non-private heterogeneous coded caching scheme into a
demand-private scheme in
Subsection~\ref{subsec:general_drs}.
We then establish the exact demand-private memory--rate
trade-off for the two-user, two-file setting ($N=K=2$)
in Subsection~\ref{subsec:exact_nk2}.

\subsection{General Construction via Restricted Demand Subset}
\label{subsec:general_drs}

To construct a demand-private coded caching scheme with
$N$ files and $K$ users, we start with a non-private
coded caching scheme involving $N$ files and $NK$ users.
Rather than requiring the non-private scheme to serve
all possible demand vectors, it suffices to consider
a restricted demand subset, denoted by
$\mathcal D_{\mathcal{RS}}$, which generalizes the
demand subset introduced in~\cite{Kamath19}.
We formalize the corresponding class of non-private
schemes below.
\begin{definition}[$\mathcal{D}_{\mathcal{RS}}$-non-private heterogeneous scheme]
    \label{def:drs}
    Consider a non-private coded caching system with $N$ files of sizes
    $\mathbf{F} = (F_0, F_1, \dots, F_{N-1})$ and $NK$ users partitioned
    into $K$ stacks of $N$ users each. Every user in the $k$-th stack
    ($k \in [K]$) has cache size $M_k$, giving the expanded cache size
    vector
    \begin{equation}
    \mathbf{M}_{\text{exp}} = \bigl(
            \underbrace{M_0,\dots,M_0}_{N},
            \underbrace{M_1,\dots,M_1}_{N}, \dots,
            \underbrace{M_{K-1},\dots,M_{K-1}}_{N}
        \bigr).
        \label{eq:m_exp}
    \end{equation}
    Let the demand vector $\overline{d}\in[N]^{NK}$ of the $NK$ users be
    partitioned according to the $K$ stacks as
\begin{align}
    \overline{d}
    =
    \bigl(
        \overline{d}^{(0)},
        \overline{d}^{(1)},
        \ldots,
        \overline{d}^{(K-1)}
    \bigr),
\end{align}
where $\overline{d}^{(k)}\in[N]^N$ denotes the demand vector of the $N$
users in the $k$-th stack. Define the \emph{restricted demand subset}
    $\mathcal{D}_{\mathcal{RS}} \subset [N]^{NK}$ as
    \begin{align}
        \mathcal{D}_{\mathcal{RS}} = &\Bigl\{
            \overline{d} \in [N]^{NK} : \overline{d}^{(k)} \text{ is a cyclic shift of } \nonumber\\& \qquad (0, 1, \dots, N-1), \; \forall k \in [K]
        \Bigr\}.
        \label{eq:drs_set}
    \end{align}
    An $(N, NK, \mathbf{M}_{\text{exp}}, \mathbf{F}, R)$ non-private
    coded caching scheme that achieves rate $R$ for all demand vectors
    in $\mathcal{D}_{\mathcal{RS}}$ is called an
    $(N, NK, \mathbf{M}_{\text{exp}}, \mathbf{F}, R)$
    $\mathcal{D}_{\mathcal{RS}}$-non-private heterogeneous scheme.
\end{definition}

The following theorem shows that any such restricted-demand non-private
scheme immediately yields a demand-private scheme with $K$ users.

\begin{theorem}
\label{thm:drs_reduction}
    If there exists an $(N, NK, \mathbf{M}_{\text{exp}}, \mathbf{F}, R)$
    $\mathcal{D}_{\mathcal{RS}}$-non-private heterogeneous scheme, then there exists
    an $(N, K, \mathbf{M}, \mathbf{F}, R)$-private heterogeneous
    scheme.
\end{theorem}
\begin{proof}
    See Appendix~\ref{proof:thm_extnded}.
\end{proof}

%===========================================================
\subsection{Exact Characterization for the $N = K = 2$ Case}
\label{subsec:exact_nk2}

In Theorem~\ref{thm:exact_tradeoff}, we present the exact characterization of the demand-private rate-memory region for a two-user, two-file ($N=K=2$) coded caching system with heterogeneous file and cache sizes.

\begin{theorem}
\label{thm:exact_tradeoff}
For $N=K=2$, an $(N,K,\mathbf M,\mathbf F, R) = (2,2,M_0,M_1, F_0, F_1, R)$  exists for $F_0 \geq F_1$ if and
    only if
    \begin{align}
        R &\ge F_0 + F_1 - 2M_0, \tag{DP1}\\[3pt]
        R &\ge F_0 + F_1 - 2M_1, \tag{DP2}\\[3pt]
        R &\ge (3F_0 + 2F_1 - 2M_0 - M_1)/3, \tag{DP3}\\[3pt]
        R &\ge (3F_0 + 2F_1 - M_0 - 2M_1)/3, \tag{DP4}\\[3pt]
        R &\ge F_0 + 0.5F_1 - M_0, \tag{DP5}\\[3pt]
        R &\ge F_0 + 0.5F_1 - M_1, \tag{DP6}\\[3pt]
        R &\ge 0.5(F_0 + F_1 - M_0), \tag{DP7}\\[3pt]
        R &\ge 0.5(F_0 + F_1 - M_1). \tag{DP8}
    \end{align}
\end{theorem}
\begin{proof}
    See Section~\ref{sec:proofs}.
\end{proof}

\begin{corollary}
\label{cor:thm_exact_N_K_2}
For $F_1\geq F_0$, the exact demand-private
memory-rate region is obtained by exchanging
$F_0$ and $F_1$ in
\Cref{thm:exact_tradeoff}.
\end{corollary}
\begin{proof}
For $F_1\geq F_0$, relabel the files as
$W_0'=W_1$ and $W_1'=W_0$, and the demands
as $D_k'=1-D_k$ for $k\in[2]$.
This relabeling is a bijection on $[2]^2$
that preserves the uniform demand distribution,
decodability, and demand privacy, without
changing the cache capacities or transmission rate.
The result therefore follows from
\Cref{thm:exact_tradeoff}.
\end{proof}
Together, \Cref{thm:exact_tradeoff,cor:thm_exact_N_K_2}
provide a complete characterization of the
achievable memory-rate region for $N=K=2$
with arbitrary file sizes and cache capacities.

%===========================================================

\section{Proof of \Cref{thm:exact_tradeoff}}
\label{sec:proofs}

We  establish the achievability part of Theorem~\ref{thm:exact_tradeoff} in Subsection~\ref{sec:atomic_schemes}, followed by the converse in Subsection~\ref{sec:converse}.

\subsection{Achievability}
\label{sec:atomic_schemes}

By assumption of Theorem~\ref{thm:exact_tradeoff}, we have $F_0 \geq F_1$, and it suffices to consider $M_0,M_1$  such that
\begin{align}
\label{eq:cache_constrnt}
0 \leq M_0,M_1 \leq F_0+F_1.
\end{align}
For a fixed $F_0$ and $F_1$, we denote an achievable memory-rate triplet by $(M_0,M_1,R)$. For notational convenience, we denote the files $W_0$ and $W_1$ by $A$ and $B$, respectively. Since $\FA \geq \FB$, we partition the larger file $A$ as
\begin{align}
A=(\Acore,\Arem),
\end{align}
where $\Acore$ has size $\FB$ and $\Arem$ has size $\FA-\FB$.

To establish the achievability of the region characterized in
\Cref{thm:exact_tradeoff}, we first identify the vertices of the polytope
defined by the inequalities in \Cref{thm:exact_tradeoff}, together with the
cache-size constraints given by \eqref{eq:cache_constrnt}.
We then construct demand-private coded caching schemes that achieve the
vertices on the lower boundary of this polytope. The achievability of all
other points on the lower boundary then follows by memory-sharing among
these vertex schemes.

Each vertex is represented by a memory-rate triple $(M_0,M_1,R)$. The
polytope has fourteen vertices, which we partition into the following two
sets:
\begin{align}
\mathcal{V}_1
&=
\{V_1,\ldots,V_9\},\\
\mathcal{V}_2
&=
\{V_3',V_4',V_7',V_8',V_9'\}.
\end{align}
The vertices in $\mathcal{V}_2$ are obtained from their corresponding
vertices in $\mathcal{V}_1$ by exchanging the user indices. Hence, it
suffices to explicitly construct schemes for the vertices in
$\mathcal{V}_1$. Among these, the schemes achieving $V_2,V_3,V_4$, and
$V_5$ are non-trivial.

By \Cref{thm:drs_reduction}, to establish the achievability of a
memory-rate triple $(M_0,M_1,R)$ for the demand-private coded caching
problem with $N=K=2$, it suffices to construct a
\mbox{$\bigl(2,4,(M_0,M_0,M_1,M_1),(\FA,\FB),R\bigr)$}-$\mathcal{D}_{\mathcal{RS}}$-non-private heterogeneous scheme. In this four-user non-private system, the users are indexed by
$(k,j)\in\{0,1\}^2$. Users $(0,0)$ and $(0,1)$ have cache size $M_0$,
whereas users $(1,0)$ and $(1,1)$ have cache size $M_1$. For $N=K=2$,
the restricted demand set $\mathcal{D}_{\mathcal{RS}}$ consists of four
demand vectors
\begin{align}
\overline{d}
=
(d_{0,0},d_{0,1},d_{1,0},d_{1,1}),
\end{align}
which we index by the demand pattern
\begin{align}
\overline{c}
=
\overline{c}(\overline{d})
\triangleq
(d_{0,0},d_{1,0})
\in\{0,1\}^2.
\end{align}
The four corresponding demand vectors are listed in
Table~\ref{tab:demand_vector}.
\begin{table}[h]
    \centering
    \renewcommand{\arraystretch}{1.3}
    \begin{tabular}{|c|c|c|c|c|}
        \hline
        $d_{0,0}$ & $d_{0,1}$ & $d_{1,0}$ & $d_{1,1}$ & $\overline{c}(\overline{d})$ \\
        \hline
        $0$ & $1$ & $0$ & $1$ & $(0,0)$ \\
        \hline
        $0$ & $1$ & $1$ & $0$ & $(0,1)$ \\
        \hline
        $1$ & $0$ & $0$ & $1$ & $(1,0)$ \\
        \hline
        $1$ & $0$ & $1$ & $0$ & $(1,1)$ \\
        \hline
    \end{tabular}
    \caption{Restricted Demand Set for $N=2$ and $K=4$.}
    \label{tab:demand_vector}
\end{table}

To establish the achievability of a vertex
\begin{align}
V=(M_0,M_1,R)
\in\mathcal{V}_1\cup\mathcal{V}_2,
\end{align}
we construct a
\mbox{$
\bigl(2,4,(M_0,M_0,M_1,M_1),(\FA,\FB),R\bigr)$}-$\mathcal{D}_{\mathcal{RS}}$-non-private heterogeneous scheme. We first
consider the scheme corresponding to the vertex
\begin{align}
V_3 = \bigl(0.5\,\FB,\; \FB,\; \FA\bigr)  \text{ and } V_3' = \bigl(\FB,\; 0.5\,\FB,\; \FA\bigr)
\end{align}

%---------------------------------------------------------------------

\subsubsection{Scheme at $V_3 = \bigl(0.5\,\FB,\; \FB,\; \FA\bigr)$ and $V_3' = \bigl(\FB,\; 0.5\,\FB,\; \FA\bigr)$} First, we discuss the caching scheme.

\noindent\underline{Caching:} Partition $\Acore$ and $B$ into two equal-size subfiles:
\begin{align}
\Acore &= (A_1,A_2),\\
B &= (B_1,B_2),
\end{align}
where each subfile has size $0.5\,\FB$. Thus,
\begin{align}
A=(A_1,A_2,\Arem).
\end{align}
The cache placement is specified in
Table~\ref{tab:cache_V3}, where $C_{k,j}$ denotes the cache content of user
$(k,j)$. Since users $(0,0)$ and $(0,1)$ cache one coded block of size
$0.5\,\FB$ and users $(1,0)$ and $(1,1)$ cache two coded blocks of size
$0.5\,\FB$, we obtain
\begin{align}
M_0=0.5\,\FB, \qquad M_1=\FB.
\end{align}

\begin{table}[]
    \centering
    \renewcommand{\arraystretch}{1.4}
    \begin{tabular}{|c|c|}
        \hline
        User $(k,j)$ & Cache content $C_{k,j}$ \\
        \hline
        $(0,0)$ & $A_1 \oplus B_1$ \\
        \hline
        $(0,1)$ & $A_2 \oplus B_2$ \\
        \hline
        $(1,0)$ & $A_1 \oplus A_2,\; B_1 \oplus B_2$ \\
        \hline
        $(1,1)$ & $A_1 \oplus A_2,\; B_1 \oplus B_2$ \\
        \hline
    \end{tabular}
    \caption{Caching of $K=4$ users at $V_3$.}
    \label{tab:cache_V3}
\end{table}

\noindent\underline{Delivery:}
For each demand pattern
$\overline c\in\{0,1\}^2$, let $T_{\overline c}$ denote the transmission
corresponding to the core files. In addition to $T_{\overline c}$, the
server transmits $\Arem$ uncoded for every demand pattern. The core
transmissions $T_{\overline c}$ are specified in Table~\ref{tab:V3_Tx}.
\begin{table}[ht]
    \centering
    \renewcommand{\arraystretch}{1.4}
    \begin{tabular}{|c|c|}
        \hline
        Demand pattern $\overline{c}$ & Transmission $T_{\overline{c}}$ \\
        \hline
        $(0,0),\; (0,1)$ & $A_2,\; B_1$ \\
        \hline
        $(1,0),\; (1,1)$ & $A_1,\; B_2$ \\
        \hline
    \end{tabular}
    \caption{Delivery to achieve vertex $V_3$.}
        \label{tab:V3_Tx}
\end{table}
Each $T_{\overline c}$ consists of two blocks, each of size $0.5\,\FB$.
Together with the uncoded transmission of $\Arem$, whose size is
$\FA-\FB$, the resulting rate is
\begin{align}
R
&=
2(0.5\,\FB)+(\FA-\FB)\\
&=
\FA.
\end{align}

\noindent\underline{Decodability:}
Consider first the demand pattern $\overline c=(0,0)$, corresponding to
\begin{align}
\overline d=(0,1,0,1).
\end{align}
In this case,
\begin{align}
T_{(0,0)}
=
\{A_2,B_1\},
\end{align}
and $\Arem$ is transmitted uncoded.
\begin{itemize}
    \item User $(0,0)$ requests $A$. Using the transmitted block $B_1$ together
    with its cached block $A_1\oplus B_1$, it recovers $A_1$. It obtains
    $A_2$ directly from the transmission. Together with the uncoded $\Arem$, it
    therefore recovers the entire file $A$.

    \item User $(0,1)$ requests $B$. Using the transmitted block $A_2$ together
    with its cached block $A_2\oplus B_2$, it recovers $B_2$. Since $B_1$
    is transmitted directly, it can recover the entire file $B$.

    \item User $(1,0)$ requests $A$. Using the transmitted block $A_2$ together
    with its cached block $A_1\oplus A_2$, it recovers $A_1$. Together with
    the transmitted block $A_2$ and the uncoded $\Arem$, it recovers the
    entire file $A$.

    \item Finally, user $(1,1)$ requests $B$. Using the transmitted block $B_1$
    together with its cached block
    \begin{align}
    B_1\oplus B_2
    \end{align}
    it recovers $B_2$. Together with the transmitted block $B_1$,
    this yields the entire file $B$.
\end{itemize}

Swapping the roles of User~$0$ and User~$1$ yields the scheme achieving the symmetric point $V_3'$.

The above decoding procedure is summarized in Appendix~\ref{app:achievability_schemes}. The decodability for $\overline{c}=(0,1)$ is also verified there. The remaining two demand patterns follow by symmetry. This completes the description of the scheme achieving $V_3$.

The remaining vertices are listed in Table~\ref{tab:vertices}. The specific caching and delivery schemes achieving these remaining corner points are detailed in Appendix~\ref{app:achievability_schemes}.

\begin{table*}[t]
    \centering
    \begin{tabular}{|l|l|l|}
        \hline
        \multicolumn{3}{|c|}{\textbf{Corner Points $(M_0,\, M_1,\, R)$ of the Demand-Private Capacity Region ($\FA \ge \FB$)}} \\
        \hline
        \multicolumn{1}{|c|}{\textbf{Non-Trivial Corner Points}} &
        \multicolumn{1}{c|}{\textbf{Trivial Corner Points}} &
        \multicolumn{1}{c|}{\textbf{Symmetric Counterparts}} \\
        \hline
        $\mathbf{V_2}:\; \bigl(\tfrac{1}{3}\FB,\; \tfrac{1}{3}\FB,\; \FA + \tfrac{1}{3}\FB\bigr)$
        & $\mathbf{V_1}:\; \bigl(0,\; 0,\; \FA + \FB\bigr)$
        & $\mathbf{V_3'}:\; \bigl(\FB,\; 0.5\FB,\; \FA\bigr)$ \\
        $\mathbf{V_3}:\; \bigl(0.5\FB,\; \FB,\; \FA\bigr)$
        & $\mathbf{V_6}:\; \bigl(\FA + \FB,\; \FA + \FB,\; 0\bigr)$
        & $\mathbf{V_4'}:\; \bigl(\FA + 0.5\FB,\; \FA,\; 0.5\FB\bigr)$ \\
        $\mathbf{V_4}:\; \bigl(\FA,\; \FA + 0.5\FB,\; 0.5\FB\bigr)$
        & $\mathbf{V_7}:\; \bigl(0,\; \FA + \FB,\; \FA + \FB\bigr)$
        & $\mathbf{V_7'}:\; \bigl(\FA + \FB,\; 0,\; \FA + \FB\bigr)$ \\
        $\mathbf{V_5}:\; \bigl(\FA + \tfrac{1}{3}\FB,\; \FA + \tfrac{1}{3}\FB,\; \tfrac{1}{3}\FB\bigr)$
        & $\mathbf{V_8}:\; \bigl(0.5\FB,\; \FA + \FB,\; \FA\bigr)$
        & $\mathbf{V_8'}:\; \bigl(\FA + \FB,\; 0.5\FB,\; \FA\bigr)$ \\
        & $\mathbf{V_9}:\; \bigl(\FA,\; \FA + \FB,\; 0.5\FB\bigr)$
        & $\mathbf{V_9'}:\; \bigl(\FA + \FB,\; \FA,\; 0.5\FB\bigr)$ \\
        \hline
    \end{tabular}
    \caption{All vertices}
    \label{tab:vertices}
\end{table*}
%===========================================================
%===========================================================

\subsubsection{Achievability via Memory Sharing among the Vertices}
\label{sec:convex_hull_synthesis}

The schemes constructed above establish the achievability of all vertices in
$\mathcal{V}_1\cup\mathcal{V}_2$. We now use memory-sharing among these
 schemes to establish achievability for every pair of memory
\begin{align}
(M_0,M_1)\in[0,\FA+\FB]^2.
\end{align}

For this purpose, partition the cache-size domain $[0,\FA+\FB]^2$ into eight
polygonal regions $\mathcal{R}_1,\ldots,\mathcal{R}_8,$
where $\mathcal{R}_i$ is the region over which the $i$-th lower bound in
\Cref{thm:exact_tradeoff} is active. More precisely, if
$f_i(M_0,M_1)$ denotes the right-hand side of $(DP_i)$, define
\begin{align}
\label{eq:def_region}
\mathcal{R}_i
\triangleq
\Bigl\{&
(M_0,M_1)\in[0,\FA+\FB]^2: \nonumber\\
& \qquad f_i(M_0,M_1)
\geq
f_j(M_0,M_1), \;
\forall j \neq i
\Bigr\}.
\end{align}
The regions $\mathcal{R}_1,\ldots,\mathcal{R}_8$ and the vertices of the
corresponding facets are listed in \Cref{tab:dp_convex_hull}.

For every $(M_0,M_1)\in\mathcal{R}_i$, the boundary rate is
\begin{align}
R=f_i(M_0,M_1).
\end{align}
For fixed $(\FA,\FB)$, $f_i$ is an affine function of $(M_0,M_1)$.
Consequently, the graph
\begin{align}
\bigl\{
(M_0,M_1,R):
(M_0,M_1)\in\mathcal{R}_i,\;
R=f_i(M_0,M_1)
\bigr\}
\end{align}
forms a facet of the lower boundary of the achievable rate-memory region.
We refer to this facet as $DP_i$. All regions $\mathcal{R}_i$ and their corresponding achievable rates are illustrated in Fig.~\ref{fig:capacity_map_2d}.

For a vertex $v_i=(M_0, M_1, R)$ of facet $DP_i$, let $v^m_i$ denote its
projection onto the $(M_0, M_1)$-plane, i.e.,
\begin{align}
    v^m_i = (M_0, M_1).
\end{align}
For every vertex $v_i=(M_0, M_1, R)$ of $DP_i$, its projection satisfies
$v^m_i\in \mathcal{R}_i$ and its rate $R$ satisfies
\begin{align}
R=f_i(M_0,M_1).
\end{align}
Moreover, $\mathcal{R}_i$ is the convex hull of the projections $v^m_i$ of
the vertices of $DP_i$. Hence, every $(M_0,M_1)\in\mathcal{R}_i$ can be
expressed as a convex combination of these projected vertices $v^m_i$.
Memory-sharing among the corresponding vertex schemes therefore achieves
the rate
\begin{align}
R=f_i(M_0,M_1).
\end{align}

\begin{claim}
\label{claim:region_partition}
The regions $\mathcal{R}_1,\ldots,\mathcal{R}_8$ defined in \eqref{eq:def_region}
form a partition of the cache-size domain
$[0,\FA+\FB]^2$.
Specifically,
\begin{align}
\bigcup_{i=1}^{8}\mathcal{R}_i
&=[0,\FA+\FB]^2,
\end{align}
and
\begin{align}
\operatorname{int}(\mathcal{R}_i)
\cap
\operatorname{int}(\mathcal{R}_j)
&=\emptyset,
\qquad \forall i\neq j.
\end{align}
\end{claim}
\begin{proof}
See Appendix~\ref{app:partition_proof}.
\end{proof}
\vspace{1mm}
By Claim~\ref{claim:region_partition}, the cache-size domain is perfectly tiled. On any boundary shared by two or more regions, the corresponding active functions $f_i$ take the same value. Therefore, for every cache-size pair $(M_0,M_1)\in[0,\FA+\FB]^2$, the boundary rate
\begin{align}
R
=
\max_{i\in\{1,\dots,8\}} f_i(M_0,M_1)
\end{align}
is achievable by memory sharing. Since any rate $R' \geq R$ is also achievable by definition, this completes the proof of the achievability part of \Cref{thm:exact_tradeoff}.

% %  C. Converse Proofs
% % =====================================================================

\subsection{Proof of Converse}
\label{sec:converse}
\par\medskip\noindent

We now establish the converse bounds (DP1)--(DP8) of \Cref{thm:exact_tradeoff}. Bounds (DP7) and (DP8) follow directly from the non-private converse in~\cite{Chang22}, while bounds (DP3) and (DP4) extend the homogeneous demand-private converse in~\cite{Gurjarpadhye22} by explicitly accounting for unequal file sizes $(F_0,F_1)$. To derive latter bounds, we use the following lemma, originally established in~\cite{KamathRD20}, which continues to hold in the heterogeneous setting.

\begin{lemma}[\cite{Gurjarpadhye22}]
    \label{lem:converse}
    Let $\tilde{i} = (i+1) \bmod 2$ for $i=0,1$. Then any demand-private scheme for $N=K=2$ satisfies the following for $j=0,1$:
    \begin{align}
        (X, Z_1, W_{j} \mid D_1=j) &\sim (X, Z_1, W_{j} \mid D_0=i, D_1=j) \notag\\
        &\sim (X, Z_1, W_{j} \mid D_0 = \tilde{i}, D_1=j) \label{Eq_priv_cond3}
    \end{align}
    and
    \begin{align} 
        (X, Z_0, W_{j} \mid D_0=j) &\sim (X, Z_0, W_{j} \mid D_0=j, D_1=i) \notag\\
        &\sim (X, Z_0, W_{j} \mid D_0 = j, D_1=\tilde{i}). \label{Eq_priv_cond4}
    \end{align}
\end{lemma}
\noindent
The proof of Lemma~\ref{lem:converse} in~\cite{Gurjarpadhye22} relies only on the recovery condition \eqref{Eq_dec_cond_hetero} and the demand-privacy condition \eqref{Eq_priv_cond_hetero} for fixed $j \in \{0,1\}$, and therefore holds unchanged for heterogeneous file sizes $(F_0, F_1)$ and cache sizes $(M_0, M_1)$.

%---------------------------------------------------------------------
\subsubsection{Proof of (DP1) and (DP2)}

Under the demand vector $(0,1)$, User~0 must decode $W_0$. Hence,
\begin{align}
    &H(W_0 \mid X, D_0=0, D_1=1) \notag \\
    &\le H(W_0, Z_0 \mid X, D_0=0, D_1=1) \label{eq:dp1_add_z}\\
    &= H(Z_0 \mid X, D_0=0, D_1=1) \notag \\
    &\quad + H(W_0 \mid X, Z_0, D_0=0, D_1=1) \notag\\
    &= H(Z_0 \mid X, D_0=0, D_1=1) \label{eq:dp1_decod2}\\
    &\le H(Z_0 \mid D_0=0, D_1=1) \notag\\
    &\le M_0. \label{eq:dp1_fano1}
\end{align}
Here, \eqref{eq:dp1_decod2} follows from the decodability constraint,
since $W_0$ can be recovered from $(X, Z_0)$ when $D_0=0, D_1=1$.

Similarly, under the demand vector $(1,1)$, User~0 must decode $W_1$.
Therefore,
\begin{align}
    H(W_1 \mid X, D_0=1, D_1=1)
    \le & H(Z_0 \mid D_0=1, D_1=1)\notag \\
    \le & M_0. \label{eq:dp1_fano2}
\end{align}

Now, by Lemma~\ref{lem:converse}, with $j=1$, we have
\begin{align}
    (X, Z_1, W_1 \mid D_0=0, D_1=1)
    \sim
    (X, Z_1, W_1 \mid D_0=1, D_1=1).
\end{align}
Marginalizing over $Z_1$ yields
\begin{align}
    (X, W_1 \mid D_0=0, D_1=1)
    \sim
    (X, W_1 \mid D_0=1, D_1=1).
\end{align}
Consequently,
\begin{align}
    H(W_1 \mid X, D_0=0, D_1=1)
    &= H(W_1 \mid X, D_0=1, D_1=1) \notag\\
    &\le M_0. \label{eq:dp1_substitution}
\end{align}

We can now bound the total entropy of the two independent files as
\begin{align}
    F_0 + F_1
    &= H(W_0, W_1) \notag\\
    &= H(W_0, W_1 \mid D_0=0, D_1=1) \notag\\
    &\le H(W_0, W_1, X \mid D_0=0, D_1=1) \notag\\
    &= H(X \mid D_0=0, D_1=1) \notag\\
    &\quad + H(W_0, W_1 \mid X, D_0=0, D_1=1) \notag\\
    &\le H(X \mid D_0=0, D_1=1) \notag\\
    &\quad + H(W_0 \mid X, D_0=0, D_1=1) \notag\\
    &\quad + H(W_1 \mid X, D_0=0, D_1=1) \\
    &\le R + 2M_0,
\end{align}
where the last inequality follows from
\eqref{eq:dp1_fano1}, \eqref{eq:dp1_substitution}, and from the fact that $H(X \mid D_0=0, D_1=1) \leq R$. Therefore,
\begin{align}
    R \ge F_0 + F_1 - 2M_0. \tag{DP1}
\end{align}

By an analogous argument with the roles of Users~0 and~1 interchanged,
we obtain
\begin{align}
    R \ge F_0 + F_1 - 2M_1. \tag{DP2}
\end{align}
%---------------------------------------------------------------------
\subsubsection{Proof of (DP3) and (DP4)}
 To obtain the bound $3R + M_0 + 2M_1 \ge 3F_0 + 2F_1$, we use the
    distributional equivalence from Lemma~\ref{lem:converse}. We first lower bound the quantity
    \begin{align}
     \mathsf{H}& = H(Z_0, X \mid D_0{=}0) + H(Z_1, X \mid D_1{=}0) \\
        &\qquad + H(Z_1, X \mid \dem{1}{0}) \label{eq:converse_DP3}
    \end{align}
    by $3F_0 + 2F_1$, and then upper bound the same quantity by
    $M_0 + 2M_1 + 3R$.

    We first lower bound the sum of the first two terms of $\mathsf{H}$:
    \begin{align}
        & H(Z_0, X \mid D_0{=}0) + H(Z_1, X \mid D_1{=}0) \notag\\
        &= H(Z_0, W_0, X \mid D_0{=}0) \notag\\
        &\qquad + H(Z_1, W_0, X \mid D_1{=}0) \label{eq:decodability1}\\
        &= H(Z_0, W_0, X \mid \dem{0}{0}) \notag\\
        &\qquad + H(Z_1, W_0, X \mid \dem{0}{0}) \label{eq:lemma_apply1}\\
        &\ge H(Z_0, Z_1, W_0, X \mid \dem{0}{0}) \notag\\
        &\qquad + H(W_0, X \mid \dem{0}{0}), \label{eq:lower1}
    \end{align}
    where \eqref{eq:decodability1} follows from decodability, 
    \eqref{eq:lemma_apply1} follows from Lemma~\ref{lem:converse}, and \eqref{eq:lower1} follows from the submodularity of entropy. Thus it follows that 
    \begin{align}
    \mathsf{H} &\geq H(Z_0, Z_1, W_0, X \mid \dem{0}{0}) \notag\\
        &\qquad + H(W_0, X \mid \dem{0}{0}) + H(Z_1, X \mid \dem{1}{0}) \label{eq:lower_H}
    \end{align}
        Let us denote the sum of the first and the third terms on the RHS of \eqref{eq:lower_H} by $\mathsf{H'}$. Then it follows that
       \begin{align}
    \mathsf{H} &\geq  \mathsf{H'} + H(W_0, X \mid \dem{0}{0}) \label{eq:lower2_H}
    \end{align}     

    Now consider $\mathsf{H'}$
    \begin{align}
       \mathsf{H'} & =  H(Z_0, Z_1, W_0, X \mid \dem{0}{0}) \notag\\
        &\qquad + H(X, Z_1 \mid \dem{1}{0}) \notag\\
        &= H(Z_0, Z_1, W_0, X \mid \dem{0}{0}) \notag\\
        &\qquad + H(X, Z_1, W_0 \mid \dem{1}{0}) \notag\\
        &\ge H(Z_0, Z_1, W_0 \mid \dem{0}{0}) \notag\\
        &\qquad + H(X, Z_1, W_0 \mid \dem{1}{0}) \notag\\
        &= H(Z_0, Z_1, W_0 \mid \dem{1}{0}) \notag\\
        &\qquad + H(X, Z_1, W_0 \mid \dem{1}{0}) \label{eq:lemma_apply2}\\
        &\ge H(Z_0, X \mid Z_1, W_0, \dem{1}{0}) \notag\\
        &\qquad + 2H(Z_1, W_0 \mid \dem{1}{0}) \label{eq:decodability2}\\
        &\ge H(W_1, Z_1, W_0 \mid \dem{1}{0}) \notag\\
        &\qquad + H(Z_1, W_0 \mid \dem{1}{0}) \notag\\
        &\ge H(W_1, W_0) + H(Z_1, W_0 \mid \dem{1}{0}). \label{eq:lower2}
    \end{align}

    It thus follows from \eqref{eq:lower2_H} and \eqref{eq:lower2} that 
    \begin{align}
  \mathsf{H}
        &\ge H(W_1, W_0) + H(Z_1, W_0 \mid \dem{1}{0}) \notag\\
        &\qquad + H(W_0, X \mid \dem{0}{0}) \notag\\
        &= H(W_1, W_0) + H(Z_1, W_0 \mid \dem{0}{1}) \notag\\
        &\qquad + H(W_0, X \mid \dem{0}{1}) \label{eq:lemma_apply3}\\
        &\ge H(W_1, W_0) \notag\\
        &\qquad + H(Z_1, X, W_1 \mid W_0, \dem{0}{1}) \notag\\
        &\qquad + 2H(W_0 \mid \dem{0}{1}) \label{eq:decodability3}\\
        &\ge H(W_1, W_0) + H(W_1 \mid W_0, \dem{0}{1}) \notag\\
        &\qquad + 2H(W_0 \mid \dem{0}{1}) \notag\\
        &= (F_0 + F_1) + F_1 + 2F_0 \notag\\
        &= 3F_0 + 2F_1, \label{eq:final_lower}
    \end{align}
    where \eqref{eq:lemma_apply2} and \eqref{eq:lemma_apply3} follow from
    Lemma~\ref{lem:converse}, and \eqref{eq:decodability2} and \eqref{eq:decodability3}
    follow from decodability together with the fact that conditioning
    reduces entropy.

Now we obtain an  upper bound for $\mathsf{H}$ given by
    \begin{align}
      \mathsf{H} =  & H(Z_0, X \mid D_0{=}0) + H(Z_1, X \mid D_1{=}0) \notag\\
        &\qquad + H(Z_1, X \mid \dem{1}{0}) \notag\\
        &\le H(Z_0 \mid D_0{=}0) + H(X \mid D_0{=}0) \notag\\
        &\qquad + H(X \mid D_1{=}0) + H(Z_1 \mid D_1{=}0) \notag\\
        &\qquad + H(Z_1 \mid \dem{1}{0}) \notag\\
        &\qquad + H(X \mid \dem{1}{0}) \notag\\
        &\le M_0 + 3R + 2M_1, \label{eq:final_upper}
    \end{align}
    where the last inequality holds because, by the privacy condition, the
    cache contents and the transmission must be independent of the users'
    demands.

    From \eqref{eq:final_lower} and \eqref{eq:final_upper} we can conclude that
    \[
        M_0 + 2M_1 + 3R \ge 3F_0 + 2F_1,
    \]
    which rearranges to (DP4). By symmetry, swapping the cache indices gives
    (DP3).

%---------------------------------------------------------------------

\subsubsection{Proof of (DP5) and (DP6)}
To obtain the bound $2R + 2M_0 \ge 2F_0 + F_1$, we use the
    distributional equivalence from Lemma~\ref{lem:converse}. We first
    lower bound the quantity
    \[
        2H(Z_0, X \mid \dem{0}{0})
    \]
    by $2F_0 + F_1$, and then upper bound the same quantity by
    $2M_0 + 2R$.

    Under $(\dem{0}{0})$, User~0 decodes $W_0$ from $(Z_0, X)$. We first
    lower bound $2H(Z_0, X \mid \dem{0}{0})$ as follows:
    \begin{align}
        & 2H(Z_0, X \mid \dem{0}{0}) \notag\\
        &= 2H(Z_0, W_0, X \mid \dem{0}{0}) \label{eq:dp5_dec1}\\
        &\ge H(Z_0, W_0 \mid \dem{0}{0}) \notag\\
        &\qquad + H(W_0, X \mid \dem{0}{0}) \label{eq:dp5_drop}\\
        &= H(Z_0, W_0 \mid \dem{1}{0}) \notag\\
        &\qquad + H(W_0, X \mid \dem{1}{0}), \label{eq:dp5_lemma}
    \end{align}
    where \eqref{eq:dp5_dec1} follows from decodability of $W_0$,
    \eqref{eq:dp5_drop} follows because conditioning reduces entropy, and
    \eqref{eq:dp5_lemma} follows from Lemma~\ref{lem:converse} together
    with the independence of $(Z_0, W_0)$ from the demands.

    Now we find a lower bound on the right-hand side of
    \eqref{eq:dp5_lemma}:
    \begin{align}
        & H(Z_0, W_0 \mid \dem{1}{0}) + H(W_0, X \mid \dem{1}{0}) \notag\\
        &= H(Z_0 \mid W_0, \dem{1}{0}) + H(X \mid W_0, \dem{1}{0}) \notag\\
        &\qquad + 2H(W_0 \mid \dem{1}{0}) \label{eq:dp5_chain}\\
        &\ge H(Z_0, X \mid W_0, \dem{1}{0}) \notag\\
        &\qquad + 2H(W_0 \mid \dem{1}{0}) \label{eq:dp5_cond}\\
        &= H(Z_0, X, W_1 \mid W_0, \dem{1}{0}) \notag\\
        &\qquad + 2H(W_0 \mid \dem{1}{0}) \label{eq:dp5_dec2}\\
        &\ge H(W_1 \mid W_0, \dem{1}{0}) \notag\\
        &\qquad + 2H(W_0 \mid \dem{1}{0}) \notag\\
        &= F_1 + 2F_0 \notag\\
        &= 2F_0 + F_1, \label{eq:dp5_lower_final}
    \end{align}
    where \eqref{eq:dp5_chain} is the chain rule of entropy,
    \eqref{eq:dp5_cond} holds because conditioning reduces entropy, and
    \eqref{eq:dp5_dec2} follows from decodability of $W_1$ from
    $(Z_0, X)$ under $(\dem{1}{0})$.

Now we find an upper bound to $  2H(Z_0, X \mid \dem{0}{0})$
as follows:
    \begin{align}
        2H(Z_0, X \mid \dem{0}{0})
        &\le 2H(Z_0 \mid \dem{0}{0}) \notag\\
        &\qquad + 2H(X \mid \dem{0}{0}) \notag\\
        &\le 2M_0 + 2R. \label{eq:dp5_upper_final}
    \end{align}

    From \eqref{eq:dp5_lemma}, \eqref{eq:dp5_lower_final}, and
    \eqref{eq:dp5_upper_final} we obtain
    \[
        2M_0 + 2R \ge 2F_0 + F_1,
    \]
    which rearranges to
    \begin{equation}
        R \ge F_0 + 0.5F_1 - M_0. 
    \end{equation}
    which establishes (DP5).
    We obtain (DP6) by symmetry, swapping the cache indices 
    \begin{equation}
        R \ge F_0 + 0.5F_1 - M_1. 
    \end{equation}
This completes the converse proof and, together with the achievability result, establishes Theorem~\ref{thm:exact_tradeoff}.

%\input{results2}

%===========================================================
\section{Discussion}
\label{sec:discussion}
%===========================================================

The exact memory--rate trade-off characterized in
\Cref{thm:exact_tradeoff} reveals a cache-saturation
phenomenon in the presence of heterogeneous cache sizes.
Suppose that User~$0$ has the smaller cache, i.e.,
$M_0\leq M_1$. For a fixed $M_0$, increasing $M_1$
reduces the optimal transmission rate only up to a
saturation threshold $M_1^{\mathrm{sat}}(M_0)$,
beyond which additional cache memory at User~$1$
provides no further rate reduction.
This behavior is captured by the regions
$\mathcal{R}_1$, $\mathcal{R}_5$, and
$\mathcal{R}_7$, where the optimal rate depends
only on $M_0$.
The boundaries of these regions determine the
piecewise-linear saturation threshold, as illustrated
in \cref{fig:capacity_map_2d}. Interestingly, the saturation threshold is
non-monotonic in $M_0$. As $M_0$ increases from
$\FB/3$ to $\FA$, the threshold increases,
reaching $\FA+\FB/2$ at vertex $V_4$.
Beyond $M_0=\FA$, the threshold decreases
until it meets the diagonal $M_0=M_1$ at $V_5$.
Thus, increasing the smaller cache does not
necessarily reduce the amount of cache memory
required at the other user to achieve saturation.

The achievability constructions provide further
insight into this phenomenon. At the asymmetric
vertices $V_3$ and $V_4$, the delivery strategy
can be designed around the demand of the
smaller-cache user, while the larger-cache user
exploits its additional cached information to
recover its requested file from the same broadcast.
Consequently, the larger cache does not further
reduce the transmission rate.
In contrast, the constructions at $V_2$ and $V_5$
employ delivery strategies that jointly account
for the demands of both users.

Comparing the optimal demand-private trade-off
with its non-private counterpart reveals the
transmission-rate penalty imposed by demand privacy.
In particular, six of the eight defining bounds,
namely (DP1)--(DP6), characterize constraints
arising in the demand-private setting, whereas
(DP7) and (DP8) coincide with the corresponding
non-private converse bounds.
The latter two bounds are active in regions
$\mathcal{R}_7$ and $\mathcal{R}_8$, respectively,
where demand privacy can therefore be guaranteed
without an additional transmission-rate penalty.
Among the six demand-private bounds,
(DP1), (DP2), (DP5), and (DP6) are new converse
bounds specific to the heterogeneous setting,
while (DP3) and (DP4) extend the homogeneous
demand-private converse bounds to unequal
file and cache sizes.

Finally, we note that the schemes obtained through
\Cref{thm:drs_reduction} satisfy a stronger
secrecy requirement than conventional demand
privacy. Specifically, for every $k\in[K]$,
\begin{align}
I\bigl(
\overline D_{-k};
Z_k,X_{\overline D},D_k,\overline W
\bigr)=0.
\end{align}
Thus, demand privacy is preserved even when
a user is additionally given access to the
entire file library.

%\section{Plots}
%\label{sec:plots}
\input{plots.tex}

\appendices

\onecolumn
\section{Proof of Theorem~\ref{thm:drs_reduction}}
\label{proof:thm_extnded}

The proof extends the virtual-user construction of \cite{Gurjarpadhye22} to the fully heterogeneous regime. Let the $(N, NK, \mathbf{M}_{\text{exp}}, \mathbf{F}, R)$ $\mathcal{D}_{\mathcal{RS}}$-non-private scheme be defined by cache functions $C^{(np)}$, delivery functions $E^{(np)}$, and decoders $G^{(np)}$. We construct the $(N, K, \mathbf{M}, \mathbf{F}, R)$ demand-private scheme as follows.

\textbf{Cache Placement:} Let $S_0, \dots, S_{K-1}$ be mutually independent secret keys uniformly distributed over $[N]$. User $k \in [K]$ adopts the cache placement of the $S_k$-th virtual user in the $k$-th stack:
\begin{equation}
Z_k = \bigl(C^{(np)}_{kN+S_k}(\overline{W}), S_k\bigr).
\end{equation}
Because the $N$ virtual users in the $k$-th stack of $\mathbf{M}_{\text{exp}}$ are all assigned the exact cache capacity $M_k$, the physical cache constraint $M_k$ for user $k$ is strictly satisfied regardless of the realization of $S_k$.

\textbf{Delivery:} Let $\Psi$ denote the cyclic shift operator and $\mathbb{I} = (0, 1, \dots, N-1)$. The expanded demand vector $\overline{D}^{(np)} \in \mathcal{D}_{\mathcal{RS}}$ is constructed by applying $S_k \ominus D_k$ shifts to $\mathbb{I}$ for each stack $k$:
\begin{equation}
\overline{D}^{(np)}(\overline{D}, \overline{S}) = \bigl(\Psi^{S_0 \ominus D_0}(\mathbb{I}), \dots, \Psi^{S_{K-1} \ominus D_{K-1}}(\mathbb{I})\bigr).
\end{equation}
The server broadcasts $X_{\overline{D}} = \bigl(E^{(np)}(\overline{W}, \overline{D}^{(np)}),\; \overline{S} \ominus \overline{D}\bigr)$, achieving the identical transmission rate $R$.

\textbf{Decodability:} User $k$ extracts $\overline{D}^{(np)}$ from $\overline{S} \ominus \overline{D}$ and employs the non-private decoder $G^{(np)}_{kN+S_k}$. Because $\overline{D}^{(np)} \in \mathcal{D}_{\mathcal{RS}}$, zero-error decodability is guaranteed by the non-private scheme. The target file index corresponds to the $(kN+S_k)$-th component of $\overline{D}^{(np)}$, which simplifies to $S_k \ominus (S_k \ominus D_k) = D_k$. Thus, user $k$ perfectly recovers file $W_{D_k}$ of heterogeneous size $F_{D_k}$.

\textbf{Privacy:} The demand-privacy requirement $I(\overline{D}_{-k}; Z_k, D_k, X_{\overline{D}} \mid \overline{W}) = 0$ holds because $S_k \ominus D_k$ acts as a perfect one-time pad. Specifically, given $\overline{W}$, the tuple $(Z_k, D_k, X_{\overline{D}})$ is fully determined by $(S_k, D_k, \overline{S} \ominus \overline{D})$, yielding:
\begin{equation}
I(\overline{D}_{-k}; Z_k, D_k, X_{\overline{D}} \mid \overline{W}) = H(S_k, D_k, \overline{S} \ominus \overline{D}) - H(S_k, D_k, \overline{S} \ominus \overline{D} \mid \overline{D}_{-k}) = 0.
\end{equation}
This independence relies exclusively on the uniformity of the secret keys $S_k$ and the modular arithmetic of the demand mapping. It is strictly decoupled from the underlying file sizes $\mathbf{F}$ or cache sizes $\mathbf{M}$, ensuring that full heterogeneity does not compromise information-theoretic secrecy.

\onecolumn

\section{Achievability of vertices}
\label{app:achievability_schemes}

\subsection{Vertex $V_2$}

\par\medskip\noindent\textbf{Scheme at $V_2 = \left(\tfrac{1}{3}\,\FB,\; \tfrac{1}{3}\,\FB,\; \FA + \tfrac{1}{3}\,\FB\right)$:}

\step{File partition.}
Split $\Acore$ and $B$ into three equal subfiles of size $\tfrac{1}{3}\,\FB$ each:
\[
    \Acore = (A_1, A_2, A_3), \qquad B = (B_1, B_2, B_3),
\]
so that $A = (A_1, A_2, A_3, \Arem)$ with $|\Arem| = \FA - \FB$.

\step{Cache placement.}
Each user $(k,j)$ caches the core pieces $C_{k,j}$ shown below:
\begin{center}
    \renewcommand{\arraystretch}{1.4}
    \begin{tabular}{|c|c|}
        \hline
        User $(k,j)$ & Cache content $C_{k,j}$ \\
        \hline
        $(0,0)$ & $A_1 \oplus B_1$ \\
        \hline
        $(0,1)$ & $A_3 \oplus B_3$ \\
        \hline
        $(1,0)$ & $A_2 \oplus B_2$ \\
        \hline
        $(1,1)$ & $A_1 \oplus A_2 \oplus A_3 \oplus B_1 \oplus B_2 \oplus B_3$ \\
        \hline
    \end{tabular}
\end{center}
Every user stores one coded block of size $\tfrac{1}{3}\,\FB$, so
$M_0 = \frac{\FB}{3}$ and $M_1 = \frac{\FB}{3}$.

\step{Delivery.}
For every demand pattern $\overline{c} \in \{0,1\}^2$, the server transmits $\Arem$
uncoded along with the core pieces $T_{\overline{c}}$ shown below:
\begin{center}
    \renewcommand{\arraystretch}{1.4}
    \begin{tabular}{|c|c|}
        \hline
        Demand pattern $\overline{c}$ & Transmission $T_{\overline{c}}$ \\
        \hline
        $(0,0)$ & $B_1,\; B_2,\; A_3,\; A_1 \oplus A_2 \oplus A_3$ \\
        \hline
        $(0,1)$ & $A_2,\; A_3,\; B_1,\; B_1 \oplus B_2 \oplus B_3$ \\
        \hline
        $(1,0)$ & $B_2,\; B_3,\; A_1,\; A_1 \oplus A_2 \oplus A_3$ \\
        \hline
        $(1,1)$ & $A_1,\; A_2,\; B_3,\; B_1 \oplus B_2 \oplus B_3$ \\
        \hline
    \end{tabular}
\end{center}
In every case, the core transmission $T_{\overline{c}}$ consists of four blocks,
each of size $\tfrac{1}{3}\,\FB$. Together with the uncoded transmission of $\Arem$,
whose size is $\FA - \FB$, the resulting transmission rate is
\[
    R = 4\left(\frac{\FB}{3}\right) + (\FA - \FB) = \FA + \frac{\FB}{3}.
\]

\step{Decodability.}
Consider $\overline{c} = (0,0)$, corresponding to $\overline{d} =
(A, B, A, B)$, where $T_{(0,0)} = \{B_1, B_2, A_3, A_1 \oplus A_2 \oplus A_3\}$ and $\Arem$ is broadcast
uncoded:
\begin{itemize}
    \item \emph{User $(0,0)$ wants $A$:} Uses $B_1$ from $T_{(0,0)}$ to
          recover $A_1$ from its cached $A_1 \oplus B_1$. It then uses the recovered $A_1$, along with $A_3$ and $A_1 \oplus A_2 \oplus A_3$ from $T_{(0,0)}$ to recover $A_2$, and combines
          them with $\Arem$.
    \item \emph{User $(0,1)$ wants $B$:} Uses $A_3$ from $T_{(0,0)}$ to
          recover $B_3$ from its cached $A_3 \oplus B_3$, and combines
          it with $B_1$ and $B_2$ from $T_{(0,0)}$.
    \item \emph{User $(1,0)$ wants $A$:} Uses $B_2$ from $T_{(0,0)}$ to
          recover $A_2$ from its cached $A_2 \oplus B_2$. It then uses the recovered $A_2$, along with $A_3$ and $A_1 \oplus A_2 \oplus A_3$ from $T_{(0,0)}$ to recover $A_1$, and combines
          them with $\Arem$.
    \item \emph{User $(1,1)$ wants $B$:} Uses $B_1$ and $B_2$ from $T_{(0,0)}$ and its cached $A_1 \oplus A_2 \oplus A_3 \oplus B_1 \oplus B_2 \oplus B_3$, along with $A_1 \oplus A_2 \oplus A_3$ from $T_{(0,0)}$ to recover $B_3$, and combines
          it with $B_1$ and $B_2$ from $T_{(0,0)}$.
\end{itemize}
Decodability for $\overline{c} \in \{(0,1), (1,0), (1,1)\}$ follows in
the same way.

We verify the decodability for all four demand patterns in
$\mathcal{D}_{\mathcal{RS}}$. Recall that $\Arem$ is transmitted
uncoded for every demand pattern. The decoding of the core parts is
summarized in the following table.

\begin{center}
\renewcommand{\arraystretch}{1.5}
\begin{tabular}{|c|c|c|l|}
\hline
$\overline{c}$ & User & Demand & Decoding of the missing subfile(s) \\
\hline

\multirow{4}{*}{$(0,0)$}
& $(0,0)$ & $A$
& $A_1 = C_{0,0} \oplus B_1$, \quad $A_2 = (A_1 \oplus A_2 \oplus A_3) \oplus A_1 \oplus A_3$ \\
\cline{2-4}
& $(0,1)$ & $B$
& $B_3 = C_{0,1} \oplus A_3$ \\
\cline{2-4}
& $(1,0)$ & $A$
& $A_2 = C_{1,0} \oplus B_2$, \quad $A_1 = (A_1 \oplus A_2 \oplus A_3) \oplus A_2 \oplus A_3$ \\
\cline{2-4}
& $(1,1)$ & $B$
& $B_3 = C_{1,1} \oplus (A_1 \oplus A_2 \oplus A_3) \oplus B_1 \oplus B_2$ \\
\hline

\multirow{4}{*}{$(0,1)$}
& $(0,0)$ & $A$
& $A_1 = C_{0,0} \oplus B_1$ \\
\cline{2-4}
& $(0,1)$ & $B$
& $B_3 = C_{0,1} \oplus A_3$, \quad $B_2 = (B_1 \oplus B_2 \oplus B_3) \oplus B_1 \oplus B_3$ \\
\cline{2-4}
& $(1,0)$ & $B$
& $B_2 = C_{1,0} \oplus A_2$, \quad $B_3 = (B_1 \oplus B_2 \oplus B_3) \oplus B_1 \oplus B_2$ \\
\cline{2-4}
& $(1,1)$ & $A$
& $A_1 = C_{1,1} \oplus (B_1 \oplus B_2 \oplus B_3) \oplus A_2 \oplus A_3$ \\
\hline

\multirow{4}{*}{$(1,0)$}
& $(0,0)$ & $B$
& $B_1 = C_{0,0} \oplus A_1$ \\
\cline{2-4}
& $(0,1)$ & $A$
& $A_3 = C_{0,1} \oplus B_3$, \quad $A_2 = (A_1 \oplus A_2 \oplus A_3) \oplus A_1 \oplus A_3$ \\
\cline{2-4}
& $(1,0)$ & $A$
& $A_2 = C_{1,0} \oplus B_2$, \quad $A_3 = (A_1 \oplus A_2 \oplus A_3) \oplus A_1 \oplus A_2$ \\
\cline{2-4}
& $(1,1)$ & $B$
& $B_1 = C_{1,1} \oplus (A_1 \oplus A_2 \oplus A_3) \oplus B_2 \oplus B_3$ \\
\hline

\multirow{4}{*}{$(1,1)$}
& $(0,0)$ & $B$
& $B_1 = C_{0,0} \oplus A_1$, \quad $B_2 = (B_1 \oplus B_2 \oplus B_3) \oplus B_1 \oplus B_3$ \\
\cline{2-4}
& $(0,1)$ & $A$
& $A_3 = C_{0,1} \oplus B_3$ \\
\cline{2-4}
& $(1,0)$ & $B$
& $B_2 = C_{1,0} \oplus A_2$, \quad $B_1 = (B_1 \oplus B_2 \oplus B_3) \oplus B_2 \oplus B_3$ \\
\cline{2-4}
& $(1,1)$ & $A$
& $A_3 = C_{1,1} \oplus (B_1 \oplus B_2 \oplus B_3) \oplus A_1 \oplus A_2$ \\
\hline

\end{tabular}
\end{center}

Thus, every user can successfully recover its requested file for every demand vector in $\mathcal{D}_{\mathcal{RS}}$. Since $\Arem$ is also
transmitted uncoded, every user requesting $A$ can recover the entire
file $A$. Hence,
\[
    \left(M_0, M_1, R\right)
    =
    \left(
        \frac{\FB}{3},
        \frac{\FB}{3},
        \FA + \frac{\FB}{3}
    \right)
\]
is achievable.

\subsection{Vertex $V_3$}

\par\medskip\noindent\textbf{Scheme at $V_3 = \bigl(0.5\,\FB,\; \FB,\; \FA\bigr)$ and $V_3' = \bigl(\FB,\; 0.5\,\FB,\; \FA\bigr)$:}

\step{File partition.}
Split $\Acore$ and $B$ into two equal halves of size $0.5\,\FB$ each:
\[
    \Acore = (A_1, A_2), \qquad B = (B_1, B_2),
\]
so that $A = (A_1, A_2, \Arem)$ with $|\Arem| = \FA - \FB$.

\step{Cache placement.}
Each user $(k,j)$ caches the core pieces $C_{k,j}$ shown below:
\begin{center}
    \renewcommand{\arraystretch}{1.4}
    \begin{tabular}{|c|c|}
        \hline
        User $(k,j)$ & Cache content $C_{k,j}$ \\
        \hline
        $(0,0)$ & $A_1 \oplus B_1$ \\
        \hline
        $(0,1)$ & $A_2 \oplus B_2$ \\
        \hline
        $(1,0)$ & $A_1 \oplus A_2,\; B_1 \oplus B_2$ \\
        \hline
        $(1,1)$ & $A_1 \oplus A_2,\; B_1 \oplus B_2$ \\
        \hline
    \end{tabular}
\end{center}
Users $(0,0)$ and $(0,1)$ store one block of size $0.5\,\FB$, so
$M_0 = 0.5\,\FB$. Users $(1,0)$ and $(1,1)$ store two blocks of size
$0.5\,\FB$, so $M_1 = \FB$.

\step{Delivery.}
For every demand pattern $\overline{c}$, the server transmits $\Arem$
uncoded along with the core pieces $T_{\overline{c}}$ shown below:
\begin{center}
    \renewcommand{\arraystretch}{1.4}
    \begin{tabular}{|c|c|}
        \hline
        Demand pattern $\overline{c}$ & Transmission $T_{\overline{c}}$ \\
        \hline
        $(0,0),\; (0,1)$ & $A_2,\; B_1$ \\
        \hline
        $(1,0),\; (1,1)$ & $A_1,\; B_2$ \\
        \hline
    \end{tabular}
\end{center}
In every case, the server sends two core pieces of size $0.5\,\FB$ along
with $\Arem$ of size $\FA - \FB$, giving $R = \FB + (\FA - \FB) = \FA$.

\step{Decodability.}
Consider $\overline{c} = (0,0)$, corresponding to $\overline{d} =
(A, B, A, B)$, where $T_{(0,0)} = \{A_2, B_1\}$ and $\Arem$ is broadcast
uncoded:
\begin{itemize}
    \item \emph{User $(0,0)$ wants $A$:} Uses $B_1$ from $T_{(0,0)}$ to
          recover $A_1$ from its cached $A_1 \oplus B_1$, and combines
          it with $A_2$ from $T_{(0,0)}$ and $\Arem$.
    \item \emph{User $(0,1)$ wants $B$:} Uses $A_2$ from $T_{(0,0)}$ to
          recover $B_2$ from its cached $A_2 \oplus B_2$, and combines
          it with $B_1$ from $T_{(0,0)}$.
    \item \emph{User $(1,0)$ wants $A$:} Uses $A_2$ from $T_{(0,0)}$ to
          recover $A_1$ from its cached $A_1 \oplus A_2$, and combines
          it with $A_2$ and $\Arem$.
    \item \emph{User $(1,1)$ wants $B$:} Uses $B_1$ from $T_{(0,0)}$ to
          recover $B_2$ from its cached $B_1 \oplus B_2$, and combines
          it with $B_1$ from $T_{(0,0)}$.
\end{itemize}
Decodability for $\overline{c} \in \{(0,1), (1,0), (1,1)\}$ follows in
the same way.

We verify the decodability for all four demand patterns in
$\mathcal{D}_{\mathcal{RS}}$. Recall that $\Arem$ is transmitted
uncoded for every demand pattern. The decoding of the core parts is
summarized in the following table.

\begin{center}
\renewcommand{\arraystretch}{1.5}
\begin{tabular}{|c|c|c|c|}
\hline
$\overline{c}$ & User & Demand & Decoding of the missing subfile(s) \\
\hline

\multirow{4}{*}{$(0,0)$}
& $(0,0)$ & $A$
& $A_1=C_{0,0}\oplus B_1$ \\
\cline{2-4}
& $(0,1)$ & $B$
& $B_2=C_{0,1}\oplus A_2$ \\
\cline{2-4}
& $(1,0)$ & $A$
& $A_1=(A_1\oplus A_2)\oplus A_2$ \\
\cline{2-4}
& $(1,1)$ & $B$
& $B_2=(B_1\oplus B_2)\oplus B_1$ \\
\hline

\multirow{4}{*}{$(0,1)$}
& $(0,0)$ & $A$
& $A_1=C_{0,0}\oplus B_1$ \\
\cline{2-4}
& $(0,1)$ & $B$
& $B_2=C_{0,1}\oplus A_2$ \\
\cline{2-4}
& $(1,0)$ & $B$
& $B_2=(B_1\oplus B_2)\oplus B_1$ \\
\cline{2-4}
& $(1,1)$ & $A$
& $A_1=(A_1\oplus A_2)\oplus A_2$ \\
\hline

\multirow{4}{*}{$(1,0)$}
& $(0,0)$ & $B$
& $B_1=C_{0,0}\oplus A_1$ \\
\cline{2-4}
& $(0,1)$ & $A$
& $A_2=C_{0,1}\oplus B_2$ \\
\cline{2-4}
& $(1,0)$ & $A$
& $A_2=(A_1\oplus A_2)\oplus A_1$ \\
\cline{2-4}
& $(1,1)$ & $B$
& $B_1=(B_1\oplus B_2)\oplus B_2$ \\
\hline

\multirow{4}{*}{$(1,1)$}
& $(0,0)$ & $B$
& $B_1=C_{0,0}\oplus A_1$ \\
\cline{2-4}
& $(0,1)$ & $A$
& $A_2=C_{0,1}\oplus B_2$ \\
\cline{2-4}
& $(1,0)$ & $B$
& $B_1=(B_1\oplus B_2)\oplus B_2$ \\
\cline{2-4}
& $(1,1)$ & $A$
& $A_2=(A_1\oplus A_2)\oplus A_1$ \\
\hline

\end{tabular}
\end{center}

Thus, every user can recover the requested core file for every demand
vector in $\mathcal{D}_{\mathcal{RS}}$. Since $\Arem$ is also
transmitted uncoded, every user requesting $A$ can recover the entire
file $A$. Hence,
\[
    \left(M_0,M_1,R\right)
    =
    \left(
        \frac{\FB}{2},
        \FB,
        \FA
    \right)
\]
is achievable. The achievability of
\[
    V_3'
    =
    \left(
        \FB,
        \frac{\FB}{2},
        \FA
    \right)
\]
follows by swapping the two user stacks.

\subsection{Vertex $V_4$}

\par\medskip\noindent\textbf{Scheme at $V_4 = \bigl(\FA,\; \FA + 0.5\,\FB,\; 0.5\,\FB\bigr)$ and $V_4' = \bigl(\FA + 0.5\,\FB,\; \FA,\; 0.5\,\FB\bigr)$:}

\step{File partition.}
Split $\Acore$ and $B$ into two equal halves of size $0.5\,\FB$ each:
\[
    \Acore = (A_1, A_2), \qquad B = (B_1, B_2),
\]
with $A = (A_1, A_2, \Arem)$ and $|\Arem| = \FA - \FB$.

\step{Cache placement.}
All four users store $\Arem$ (of size $\FA - \FB$) in their caches. In
addition, each user $(k,j)$ stores the core pieces $C_{k,j}$ listed
below:
\begin{center}
    \renewcommand{\arraystretch}{1.4}
    \begin{tabular}{|c|c|}
        \hline
        User $(k,j)$ & Cache content $C_{k,j}$ \\
        \hline
        $(0,0)$ & $A_1,\; B_2$ \\
        \hline
        $(0,1)$ & $A_2,\; B_1$ \\
        \hline
        $(1,0)$ & $B_1,\; B_2,\; A_1 \oplus A_2$ \\
        \hline
        $(1,1)$ & $B_1,\; B_2,\; A_1 \oplus A_2$ \\
        \hline
    \end{tabular}
\end{center}
With $\Arem$ in every cache, the total cache sizes are
$M_0 = (\FA - \FB) + 2(0.5\,\FB) = \FA$ for users $(0,0), (0,1)$, and
$M_1 = (\FA - \FB) + 3(0.5\,\FB) = \FA + 0.5\,\FB$ for users $(1,0), (1,1)$.

\step{Delivery.}
Since all users already have $\Arem$ in their caches, $\Arem$ is never
transmitted; the server broadcasts only a single coded block of size
$0.5\,\FB$:
\begin{center}
    \renewcommand{\arraystretch}{1.4}
    \begin{tabular}{|c|c|}
        \hline
        Demand pattern $\overline{c}$ & Transmission $T_{\overline{c}}$ \\
        \hline
        $(0,0),\; (0,1)$ & $A_2 \oplus B_2$ \\
        \hline
        $(1,0),\; (1,1)$ & $A_1 \oplus B_1$ \\
        \hline
    \end{tabular}
\end{center}
Since only one block of size $0.5\,\FB$ is sent in all cases, the
achieved rate is $R = 0.5\,\FB$.

\step{Decodability.}
We verify the decodability for all four demand patterns in
$\mathcal{D}_{\mathcal{RS}}$. Recall that $\Arem$ is stored in the
cache of every user. Hence, a user requesting $A$ only needs to recover
$A_1$ and $A_2$. The decoding of the core parts is summarized in the
following table.

\begin{center}
\renewcommand{\arraystretch}{1.5}
\begin{tabular}{|c|c|c|c|}
\hline
$\overline{c}$ & User & Demand & Decoding of the missing subfile(s) \\
\hline

\multirow{4}{*}{$(0,0)$}
& $(0,0)$ & $A$
& $A_2=(A_2\oplus B_2)\oplus B_2$ \\
\cline{2-4}
& $(0,1)$ & $B$
& $B_2=(A_2\oplus B_2)\oplus A_2$ \\
\cline{2-4}
& $(1,0)$ & $A$
& $A_2=(A_2\oplus B_2)\oplus B_2$, then
  $A_1=(A_1\oplus A_2)\oplus A_2$ \\
\cline{2-4}
& $(1,1)$ & $B$
& Already has $B_1$ and $B_2$ \\
\hline

\multirow{4}{*}{$(0,1)$}
& $(0,0)$ & $A$
& $A_2=(A_2\oplus B_2)\oplus B_2$ \\
\cline{2-4}
& $(0,1)$ & $B$
& $B_2=(A_2\oplus B_2)\oplus A_2$ \\
\cline{2-4}
& $(1,0)$ & $B$
& Already has $B_1$ and $B_2$ \\
\cline{2-4}
& $(1,1)$ & $A$
& $A_2=(A_2\oplus B_2)\oplus B_2$, then
  $A_1=(A_1\oplus A_2)\oplus A_2$ \\
\hline

\multirow{4}{*}{$(1,0)$}
& $(0,0)$ & $B$
& $B_1=(A_1\oplus B_1)\oplus A_1$ \\
\cline{2-4}
& $(0,1)$ & $A$
& $A_1=(A_1\oplus B_1)\oplus B_1$ \\
\cline{2-4}
& $(1,0)$ & $A$
& $A_1=(A_1\oplus B_1)\oplus B_1$, then
  $A_2=(A_1\oplus A_2)\oplus A_1$ \\
\cline{2-4}
& $(1,1)$ & $B$
& Already has $B_1$ and $B_2$ \\
\hline

\multirow{4}{*}{$(1,1)$}
& $(0,0)$ & $B$
& $B_1=(A_1\oplus B_1)\oplus A_1$ \\
\cline{2-4}
& $(0,1)$ & $A$
& $A_1=(A_1\oplus B_1)\oplus B_1$ \\
\cline{2-4}
& $(1,0)$ & $B$
& Already has $B_1$ and $B_2$ \\
\cline{2-4}
& $(1,1)$ & $A$
& $A_1=(A_1\oplus B_1)\oplus B_1$, then
  $A_2=(A_1\oplus A_2)\oplus A_1$ \\
\hline

\end{tabular}
\end{center}

Thus, every user can recover its requested file for every demand vector
in $\mathcal{D}_{\mathcal{RS}}$. Therefore,
\[
    \left(M_0,M_1,R\right)
    =
    \left(
        \FA,
        \FA+\frac{\FB}{2},
        \frac{\FB}{2}
    \right)
\]
is achievable. The achievability of
\[
    V_4'
    =
    \left(
        \FA+\frac{\FB}{2},
        \FA,
        \frac{\FB}{2}
    \right)
\]
follows by swapping the two user stacks.

\subsection{Vertex $V_5$}

\par\medskip\noindent\textbf{Scheme at $V_5 = \bigl(\FA + \tfrac{1}{3}\FB,\; \FA + \tfrac{1}{3}\FB,\; \tfrac{1}{3}\FB\bigr)$:}

\step{File partition.}
Split $\Acore$ and $B$ into three equal subfiles of size $\FB/3$ each:
\[
    \Acore = (A_1, A_2, A_3), \qquad B = (B_1, B_2, B_3),
\]
with $A = (A_1, A_2, A_3, \Arem)$ and $|\Arem| = \FA - \FB$.

\step{Cache placement.}
All four users store $\Arem$ (of size $\FA - \FB$) in their caches. In
addition, each user $(k,j)$ caches the four core blocks $C_{k,j}$ listed
below:
\begin{center}
    \renewcommand{\arraystretch}{1.4}
    \begin{tabular}{|c|c|}
        \hline
        User $(k,j)$ & Cache content $C_{k,j}$ \\
        \hline
        $(0,0)$ & $A_2,\; A_1 \oplus A_3,\; B_1,\; B_3$ \\
        \hline
        $(0,1)$ & $A_1,\; A_3,\; B_2,\; B_1 \oplus B_3$ \\
        \hline
        $(1,0)$ & $A_3,\; A_1 \oplus A_2,\; B_1,\; B_2$ \\
        \hline
        $(1,1)$ & $A_1,\; A_2,\; B_3,\; B_1 \oplus B_2$ \\
        \hline
    \end{tabular}
\end{center}
With $\Arem$ and four blocks of size $\FB/3$ in every cache,
$M_0 = M_1 = (\FA - \FB) + \tfrac{4}{3}\FB = \FA + \tfrac{1}{3}\FB$.

\step{Delivery.}
As in the previous scheme, $\Arem$ is never transmitted; the server
broadcasts a single coded block:
\begin{center}
    \renewcommand{\arraystretch}{1.4}
    \begin{tabular}{|c|c|}
        \hline
        Demand pattern $\overline{c}$ & Transmission $T_{\overline{c}}$ \\
        \hline
        $(0,0)$ & $A_1 \oplus B_1$ \\
        \hline
        $(0,1)$ & $A_3 \oplus B_3$ \\
        \hline
        $(1,0)$ & $A_2 \oplus B_2$ \\
        \hline
        $(1,1)$ & $A_1 \oplus A_2 \oplus A_3 \oplus B_1 \oplus B_2 \oplus B_3$ \\
        \hline
    \end{tabular}
\end{center}
Every transmission has size $\FB/3$, yielding $R = \tfrac{1}{3}\FB$.

\step{Decodability.}
Consider $\overline{c} = (0,0)$, corresponding to
$\overline{d} = (A, B, A, B)$, where $T_{(0,0)} = A_1 \oplus B_1$:
\begin{itemize}
    \item \emph{User $(0,0)$ wants $A$:} Uses its cached $B_1$ to recover
          $A_1$ from $A_1 \oplus B_1$, then recovers $A_3$ from its
          cached $A_1 \oplus A_3$, and combines them with its cached
          $A_2$ and $\Arem$.
    \item \emph{User $(0,1)$ wants $B$:} Uses its cached $A_1$ to recover
          $B_1$ from $A_1 \oplus B_1$, then recovers $B_3$ from its
          cached $B_1 \oplus B_3$, and combines them with its cached
          $B_2$.
    \item \emph{User $(1,0)$ wants $A$:} Uses its cached $B_1$ to recover
          $A_1$ from $A_1 \oplus B_1$, then recovers $A_2$ from its
          cached $A_1 \oplus A_2$, and combines them with its cached
          $A_3$ and $\Arem$.
    \item \emph{User $(1,1)$ wants $B$:} Uses its cached $A_1$ to recover
          $B_1$ from $A_1 \oplus B_1$, then recovers $B_2$ from its
          cached $B_1 \oplus B_2$, and combines them with its cached
          $B_3$.
\end{itemize}
Decodability for $\overline{c} \in \{(0,1), (1,0), (1,1)\}$ follows in
the same way.

We verify the decodability for all four demand patterns in
$\mathcal{D}_{\mathcal{RS}}$. Recall that $\Arem$ is stored in the
cache of every user. Hence, a user requesting $A$ only needs to recover
the missing subfiles of $\Acore$. The decoding of the core parts is
summarized in the following table.

\begin{center}
\renewcommand{\arraystretch}{1.5}
\begin{tabular}{|c|c|c|p{8.2cm}|}
\hline
$\overline{c}$ & User & Demand & Decoding of the missing subfile(s) \\
\hline

\multirow{4}{*}{$(0,0)$}
& $(0,0)$ & $A$
& $A_1=T_{(0,0)}\oplus B_1$, then
  $A_3=(A_1\oplus A_3)\oplus A_1$ \\
\cline{2-4}
& $(0,1)$ & $B$
& $B_1=T_{(0,0)}\oplus A_1$, then
  $B_3=(B_1\oplus B_3)\oplus B_1$ \\
\cline{2-4}
& $(1,0)$ & $A$
& $A_1=T_{(0,0)}\oplus B_1$, then
  $A_2=(A_1\oplus A_2)\oplus A_1$ \\
\cline{2-4}
& $(1,1)$ & $B$
& $B_1=T_{(0,0)}\oplus A_1$, then
  $B_2=(B_1\oplus B_2)\oplus B_1$ \\
\hline

\multirow{4}{*}{$(0,1)$}
& $(0,0)$ & $A$
& $A_3=T_{(0,1)}\oplus B_3$, then
  $A_1=(A_1\oplus A_3)\oplus A_3$ \\
\cline{2-4}
& $(0,1)$ & $B$
& $B_3=T_{(0,1)}\oplus A_3$, then
  $B_1=(B_1\oplus B_3)\oplus B_3$ \\
\cline{2-4}
& $(1,0)$ & $B$
& $B_3=T_{(0,1)}\oplus A_3$ \\
\cline{2-4}
& $(1,1)$ & $A$
& $A_3=T_{(0,1)}\oplus B_3$ \\
\hline

\multirow{4}{*}{$(1,0)$}
& $(0,0)$ & $B$
& $B_2=T_{(1,0)}\oplus A_2$ \\
\cline{2-4}
& $(0,1)$ & $A$
& $A_2=T_{(1,0)}\oplus B_2$ \\
\cline{2-4}
& $(1,0)$ & $A$
& $A_2=T_{(1,0)}\oplus B_2$, then
  $A_1=(A_1\oplus A_2)\oplus A_2$ \\
\cline{2-4}
& $(1,1)$ & $B$
& $B_2=T_{(1,0)}\oplus A_2$, then
  $B_1=(B_1\oplus B_2)\oplus B_2$ \\
\hline

\multirow{4}{*}{$(1,1)$}
& $(0,0)$ & $B$
& $B_2=T_{(1,1)}\oplus A_2\oplus(A_1\oplus A_3)
  \oplus B_1\oplus B_3$ \\
\cline{2-4}
& $(0,1)$ & $A$
& $A_2=T_{(1,1)}\oplus A_1\oplus A_3
  \oplus B_2\oplus(B_1\oplus B_3)$ \\
\cline{2-4}
& $(1,0)$ & $B$
& $B_3=T_{(1,1)}\oplus A_3\oplus(A_1\oplus A_2)
  \oplus B_1\oplus B_2$ \\
\cline{2-4}
& $(1,1)$ & $A$
& $A_3=T_{(1,1)}\oplus A_1\oplus A_2
  \oplus B_3\oplus(B_1\oplus B_2)$ \\
\hline

\end{tabular}
\end{center}

Thus, every user can recover its requested core file for every demand
vector in $\mathcal{D}_{\mathcal{RS}}$. Since every user also stores
$\Arem$, every user requesting $A$ can recover the entire file $A$.
Therefore,
\[
    \left(M_0,M_1,R\right)
    =
    \left(
        \FA+\frac{\FB}{3},
        \FA+\frac{\FB}{3},
        \frac{\FB}{3}
    \right)
\]
is achievable.

\par\medskip\noindent\textbf{Trivial schemes:}

The remaining corner points lie on the outer boundary of the cache region
and are achieved either by uncoded delivery or by directly operating one
of the schemes above while ignoring the excess cache memory at the
larger-cache user:
\begin{itemize}
    \item \textbf{Point $V_1 = (0,\, 0,\, \FA + \FB)$:} Neither user has
          cache memory, and the server broadcasts both files $A$ and $B$
          uncoded.
    \item \textbf{Point $V_6 = (\FA + \FB,\, \FA + \FB,\, 0)$:} Both
          users cache both files $A$ and $B$ in full, requiring zero
          transmission ($R = 0$).
    \item \textbf{Points $V_7 = (0,\, \FA + \FB,\, \FA + \FB)$ and
          $V_7' = (\FA + \FB,\, 0,\, \FA + \FB)$:} One user has no cache
          while the other caches both files; the server broadcasts both
          files uncoded at rate $R = \FA + \FB$ so the zero-cache user
          can decode without revealing its demand.
    \item \textbf{Points $V_8 = (0.5\,\FB,\, \FA + \FB,\, \FA)$ and
          $V_8' = (\FA + \FB,\, 0.5\,\FB,\, \FA)$:} Achieved directly by
          the scheme at $V_3$ (resp.\ $V_3'$), where the user with cache
          size $\FA + \FB$ only utilizes $\FB$ of its memory (or simply
          stores both files in full).
    \item \textbf{Points $V_9 = (\FA,\, \FA + \FB,\, 0.5\,\FB)$ and
          $V_9' = (\FA + \FB,\, \FA,\, 0.5\,\FB)$:} Achieved directly by
          the scheme at $V_4$ (resp.\ $V_4'$), where the user with cache
          size $\FA + \FB$ only utilizes $\FA + 0.5\,\FB$ of its memory.
\end{itemize}
%===========================================================
\section{Proof of Complete Coverage and Disjoint Interiors of $\mathcal{R}_1, \ldots, \mathcal{R}_8$}
\label{app:partition_proof}
%===========================================================

In this appendix, we verify from the region definitions in \Cref{tab:dp_convex_hull} that the eight regions $\mathcal{R}_1, \ldots, \mathcal{R}_8$ cover the entire cache-size square $[0, \FA + \FB]^2$ and have pairwise disjoint interiors:
\begin{align}
    \bigcup_{i=1}^{8} \mathcal{R}_i &= [0, \FA + \FB]^2, \label{eq:app_union} \\
    \operatorname{int}(\mathcal{R}_i) \cap \operatorname{int}(\mathcal{R}_j) &= \emptyset, \qquad \forall i \neq j. \label{eq:app_disjoint}
\end{align}

First, we split the square $[0, \FA + \FB]^2$ along the main diagonal $M_1 = M_0$ into two triangles: the upper-left triangle $\mathcal{T}_{\mathrm{upper}}$ and the lower-right triangle $\mathcal{T}_{\mathrm{lower}}$, defined as
\begin{align}
    \mathcal{T}_{\mathrm{upper}} &\triangleq \bigl\{(M_0, M_1) \in [0, \FA + \FB]^2 : M_1 \ge M_0\bigr\}, \\
    \mathcal{T}_{\mathrm{lower}} &\triangleq \bigl\{(M_0, M_1) \in [0, \FA + \FB]^2 : M_0 \ge M_1\bigr\}.
\end{align}
Clearly, we have
\begin{equation}
    \mathcal{T}_{\mathrm{upper}} \cup \mathcal{T}_{\mathrm{lower}} = [0, \FA + \FB]^2,
\end{equation}
and their interiors are disjoint because $M_1 > M_0$ and $M_0 > M_1$ cannot hold simultaneously:
\begin{equation}
    \operatorname{int}(\mathcal{T}_{\mathrm{upper}}) \cap \operatorname{int}(\mathcal{T}_{\mathrm{lower}}) = \emptyset.
\end{equation}

From \Cref{tab:dp_convex_hull}, the four regions $\mathcal{R}_1, \mathcal{R}_3, \mathcal{R}_5, \mathcal{R}_7$ satisfy $M_1 \ge M_0$ and lie in $\mathcal{T}_{\mathrm{upper}}$, while the remaining four symmetric regions $\mathcal{R}_2, \mathcal{R}_4, \mathcal{R}_6, \mathcal{R}_8$ satisfy $M_0 \ge M_1$ and lie in $\mathcal{T}_{\mathrm{lower}}$. Therefore, any region in $\{\mathcal{R}_1, \mathcal{R}_3, \mathcal{R}_5, \mathcal{R}_7\}$ has a disjoint interior from any region in $\{\mathcal{R}_2, \mathcal{R}_4, \mathcal{R}_6, \mathcal{R}_8\}$. By symmetry, it suffices to show that $\mathcal{R}_1, \mathcal{R}_3, \mathcal{R}_5, \mathcal{R}_7$ cover $\mathcal{T}_{\mathrm{upper}}$ and have pairwise disjoint interiors.

%===========================================================
\section{Proof of Complete Coverage and Disjoint Interiors of $\mathcal{R}_1, \ldots, \mathcal{R}_8$}
\label{app:partition_proof}
%===========================================================

We verify that the regions $\mathcal{R}_1, \ldots, \mathcal{R}_8$ defined in \Cref{tab:dp_convex_hull} perfectly tile the cache-size domain $[0, \FA + \FB]^2$:
\begin{align}
    \bigcup_{i=1}^{8} \mathcal{R}_i &= [0, \FA + \FB]^2, \label{eq:app_union} \\
    \operatorname{int}(\mathcal{R}_i) \cap \operatorname{int}(\mathcal{R}_j) &= \emptyset, \qquad \forall i \neq j. \label{eq:app_disjoint}
\end{align}

%-----------------------------------------------------------
\subsection{Proof of Complete Coverage of $\mathcal{T}_{\mathrm{upper}}$}

Throughout this proof, we assume $\FA \ge \FB > 0$, as stipulated in the model.

For any point $(M_0,M_1)\in\mathcal{T}_{\mathrm{upper}}$, we have $M_0\in[0,\FA+\FB]$ and $M_1\in[M_0,\FA+\FB]$. We partition the range of $M_0$ into five intervals and verify that the entire vertical range $M_1\in[M_0,\FA+\FB]$ is covered by the regions defined in \Cref{tab:dp_convex_hull}.

\par\medskip\noindent\textbf{1) Interval $0\le M_0\le \tfrac{1}{3}\FB$:}
In this interval, $4M_0-\FB\le M_0$. Hence, every $M_1\ge M_0$ automatically satisfies $M_1\ge 4M_0-\FB$. Together with the remaining defining inequalities of $\mathcal{R}_1$, this implies:
\begin{equation}
    M_1\in[M_0,\FA+\FB] \implies (M_0,M_1)\in\mathcal{R}_1.
\end{equation}

\par\medskip\noindent\textbf{2) Interval $\tfrac{1}{3}\FB\le M_0\le 0.5\FB$:}
In this interval, $M_0\le 4M_0-\FB\le \FB\le\FA+\FB$. The boundary $M_1=4M_0-\FB$ therefore splits the vertical range into two portions. The lower portion satisfies the defining inequalities of $\mathcal{R}_3$, while the upper portion satisfies those of $\mathcal{R}_1$:
\begin{align}
    M_1\in[M_0,\,4M_0-\FB] &\implies (M_0,M_1)\in\mathcal{R}_3,\\
    M_1\in[4M_0-\FB,\,\FA+\FB] &\implies (M_0,M_1)\in\mathcal{R}_1.
\end{align}
Consequently, $[M_0,\,4M_0-\FB] \cup [4M_0-\FB,\,\FA+\FB] = [M_0,\FA+\FB]$.

\par\medskip\noindent\textbf{3) Interval $0.5\FB\le M_0\le\FA$:}
In this interval, $M_0\le M_0+0.5\FB\le\FA+\FB$. The boundary $M_1=M_0+0.5\FB$ splits the vertical range into two portions. The lower portion satisfies the defining inequalities of $\mathcal{R}_3$, while the upper portion satisfies those of $\mathcal{R}_5$:
\begin{align}
    M_1\in[M_0,\,M_0+0.5\FB] &\implies (M_0,M_1)\in\mathcal{R}_3,\\
    M_1\in[M_0+0.5\FB,\,\FA+\FB] &\implies (M_0,M_1)\in\mathcal{R}_5.
\end{align}
Consequently, $[M_0,\,M_0+0.5\FB] \cup [M_0+0.5\FB,\,\FA+\FB] = [M_0,\FA+\FB]$.

\par\medskip\noindent\textbf{4) Interval $\FA\le M_0\le\FA+\tfrac{1}{3}\FB$:}
In this interval, $M_0\le 1.5\FA+0.5\FB-0.5M_0 \le\FA+\FB$. The boundary $M_1=1.5\FA+0.5\FB-0.5M_0$ splits the vertical range into two portions. The lower portion satisfies the defining inequalities of $\mathcal{R}_3$, while the upper portion satisfies those of $\mathcal{R}_7$:
\begin{align}
    M_1\in[M_0,\, 1.5\FA+0.5\FB-0.5M_0] &\implies (M_0,M_1)\in\mathcal{R}_3,\\
    M_1\in[1.5\FA+0.5\FB-0.5M_0,\,\FA+\FB] &\implies (M_0,M_1)\in\mathcal{R}_7.
\end{align}
Consequently, $[M_0,\, 1.5\FA+0.5\FB-0.5M_0] \cup [1.5\FA+0.5\FB-0.5M_0,\,\FA+\FB] = [M_0,\FA+\FB]$.

\par\medskip\noindent\textbf{5) Interval $\FA+\tfrac{1}{3}\FB\le M_0\le\FA+\FB$:}
In this interval, $1.5\FA+0.5\FB-0.5M_0 \le M_0$. Thus, every $M_1\ge M_0$ automatically satisfies $M_1\ge 1.5\FA+0.5\FB-0.5M_0$. Together with the remaining defining inequalities of $\mathcal{R}_7$, this gives:
\begin{equation}
    M_1\in[M_0,\FA+\FB] \implies (M_0,M_1)\in\mathcal{R}_7.
\end{equation}

Combining the five intervals yields $\mathcal{R}_1\cup\mathcal{R}_3\cup\mathcal{R}_5\cup\mathcal{R}_7 =\mathcal{T}_{\mathrm{upper}}$. By symmetry, $\mathcal{R}_2\cup\mathcal{R}_4\cup\mathcal{R}_6\cup\mathcal{R}_8 =\mathcal{T}_{\mathrm{lower}}$. This proves~\eqref{eq:app_union}.

%-----------------------------------------------------------
\subsection{Proof of Pairwise Disjoint Interiors}

The interiors of the regions are characterized by their defining inequalities made strict. We first establish pairwise disjointness among $\mathcal{R}_1,\mathcal{R}_3,\mathcal{R}_5,\mathcal{R}_7$.

\par\medskip\noindent\textbf{1) $\operatorname{int}(\mathcal{R}_1) \cap\operatorname{int}(\mathcal{R}_3)=\emptyset$:}
The separating boundary is $M_1=4M_0-\FB$, since
\begin{equation}
    \operatorname{int}(\mathcal{R}_1) \implies M_1>4M_0-\FB, \quad \text{and} \quad \operatorname{int}(\mathcal{R}_3) \implies M_1<4M_0-\FB.
\end{equation}

\par\medskip\noindent\textbf{2) $\operatorname{int}(\mathcal{R}_5) \cap\operatorname{int}(\mathcal{R}_3)=\emptyset$:}
The separating boundary is $M_1=M_0+0.5\FB$, since
\begin{equation}
    \operatorname{int}(\mathcal{R}_5) \implies M_1>M_0+0.5\FB, \quad \text{and} \quad \operatorname{int}(\mathcal{R}_3) \implies M_1<M_0+0.5\FB.
\end{equation}

\par\medskip\noindent\textbf{3) $\operatorname{int}(\mathcal{R}_7) \cap\operatorname{int}(\mathcal{R}_3)=\emptyset$:}
The separating boundary is $M_1=1.5\FA+0.5\FB-0.5M_0$, since
\begin{equation}
\begin{aligned}
    \operatorname{int}(\mathcal{R}_7) &\implies M_1>1.5\FA+0.5\FB-0.5M_0,\\
    \operatorname{int}(\mathcal{R}_3) &\implies M_1<1.5\FA+0.5\FB-0.5M_0.
\end{aligned}
\end{equation}

\par\medskip\noindent\textbf{4) $\operatorname{int}(\mathcal{R}_1) \cap\operatorname{int}(\mathcal{R}_5)=\emptyset$:}
The separating boundary is $M_0=0.5\FB$, since
\begin{equation}
    \operatorname{int}(\mathcal{R}_1) \implies M_0<0.5\FB, \quad \text{and} \quad \operatorname{int}(\mathcal{R}_5) \implies M_0>0.5\FB.
\end{equation}

\par\medskip\noindent\textbf{5) $\operatorname{int}(\mathcal{R}_5) \cap\operatorname{int}(\mathcal{R}_7)=\emptyset$:}
The separating boundary is $M_0=\FA$, since
\begin{equation}
    \operatorname{int}(\mathcal{R}_5) \implies M_0<\FA, \quad \text{and} \quad \operatorname{int}(\mathcal{R}_7) \implies M_0>\FA.
\end{equation}

\par\medskip\noindent\textbf{6) $\operatorname{int}(\mathcal{R}_1) \cap\operatorname{int}(\mathcal{R}_7)=\emptyset$:}
Since $\FA\ge\FB>0$, we have $0.5\FB<\FA$. Consequently,
\begin{equation}
    \operatorname{int}(\mathcal{R}_1) \implies M_0<0.5\FB, \quad \text{and} \quad \operatorname{int}(\mathcal{R}_7) \implies M_0>\FA,
\end{equation}
which makes the two interiors disjoint.

By symmetry, the interiors of $\mathcal{R}_2,\mathcal{R}_4,\mathcal{R}_6,\mathcal{R}_8$ are also pairwise disjoint. Finally, every upper region has its interior contained in $\{(M_0,M_1):M_1>M_0\}$, whereas every lower region has its interior contained in $\{(M_0,M_1):M_0>M_1\}$. Therefore, for every $i\in\{1,3,5,7\}$ and $j\in\{2,4,6,8\}$, $\operatorname{int}(\mathcal{R}_i) \cap\operatorname{int}(\mathcal{R}_j) =\emptyset$. Combining the within-group and cross-group results proves pairwise disjointness of the interiors of all eight regions, establishing~\eqref{eq:app_disjoint}.

\bibliography{Bibliography.bib}
\bibliographystyle{IEEEtran}

\end{document}